\documentclass{article}

\usepackage[mathletters]{ucs}
\usepackage[utf8x]{inputenx}
\usepackage{amsmath,amsthm,amssymb,xspace,stmaryrd}
\usepackage{graphicx}
\usepackage[all,cmtip]{xy}
 \usepackage{comment}
 
\begin{document}

\newcommand{\drie}{\vartriangleleft}
\newcommand{\uS}{\underline{\Sigma}}
\newcommand{\uO}{\underline{\Omega}}
\newcommand{\uA}{\underline{A}}
\newcommand{\uB}{\underline{B}}
\newcommand{\topos}{[\mathcal{C}^{\text{op}},\mathbf{Set}]}
\newcommand{\idl}{\mathcal{I}(\mathcal{C})}
\newcommand{\fil}{\mathcal{F}(\mathcal{C})}
\newcommand{\Cd}{\mathcal{C}_{\downarrow}}
\newcommand{\Cu}{\mathcal{C}_{\uparrow}}
\newcommand{\SdA}{\Sigma_{A}^{\downarrow}}
\newcommand{\SdB}{\Sigma_{B}^{\downarrow}}
\newcommand{\SuA}{\Sigma_{A}^{\uparrow}}
\newcommand{\SuB}{\Sigma_{B}^{\uparrow}}

\newtheorem{dork}{Definition}[section]
\newtheorem{tut}[dork]{Theorem}
\newtheorem{poe}[dork]{Proposition}
\newtheorem{lem}[dork]{Lemma}
\newtheorem{cor}[dork]{Corollary}
\newtheorem{rem}[dork]{Remark}

\title{Topos Models for Physics and Topos Theory}
\maketitle
\begin{center}
 \author{Sander Wolters\footnote{Radboud Universiteit Nijmegen, Institute for Mathematics, Astrophysics, and Particle Physics, the Netherlands. s.wolters@math.ru.nl. Supported by
N.W.O. through project 613.000.811.}}
\end{center}
\begin{abstract}
What is the role of topos theory in the topos models for quantum theory as used by Isham, Butterfield, D\"oring, Heunen, Landsman, Spitters and others? In other words, what is the interplay between physical motivation for the models and the mathematical framework used in these models? Concretely, we show that the presheaf topos model of Butterfield, Isham and D\"oring resembles classical physics when viewed from the internal language of the presheaf topos, similar to the copresheaf topos model of Heunen, Landsman and Spitters. Both the presheaf and copresheaf models provide a `quantum logic' in the form of a complete Heyting algebra. Although these algebras are natural from a topos theoretic stance, we seek a physical interpretation for the logical operations. Finally, we investigate dynamics. In particular we describe how an automorphism on the operator algebra induces a homeomorphism (or isomorphism of locales) on the associated state spaces of the topos models, and how elementary propositions and truth values transform under the action of this homeomorphism. Also with dynamics the focus is on the internal perspective of the topos.
\end{abstract}

\section{Introduction and Motivation}

In a series of four papers~\cite{butish1, butish2, buhais,butish3}, Jeremy Butterfield and Chris Isham demonstrated that in studying foundations of quantum physics, in particular the Kochen-Specker Theorem, structures from topos theory show up in a natural way. A second series of papers~\cite{di1,di2,di3,di4} of Chris Isham, now working together with Andreas D\"oring, show greater ambition in applying topos theory to physics. A central idea in these papers is that any theory of physics, in its mathematical formulation, should share certain structures~\cite{di1}. These structures are assumed as they assist in giving some, hopefully not-naive, realist account of the theory. Aside from putting restrictions of the shape of the mathematical framework of physical theories, freedom is added in that we may use other topoi than the category of sets. We will refer to this idea as neorealism. The motivating example is the presheaf model of Butterfield and Isham, further developed by Isham and D\"oring.

This notion of neorealism raises several question. For example, what makes a good topos model? A key point is that the topos formulation of a physical theory resembles classical physics more than e.g. orthodox quantum physics.  But in what way is a topos model, in all its abstraction, closer to classical physics? One possibility to make the claim that the model `resembles classical physics' mathematically precise is to use the internal language of the topos. Indeed, it is exactly in this way that the topos model proposed by Heunen, Landsman and Spitters in~\cite{hls} resembles classical physics. This topos model uses a topos of copresheaves and is closely related to the presheaf model by Butterfield, D\"oring and Isham~\cite{wollie}. We will refer to the HLS model of copresheaves as the \textbf{covariant model} and to the BDI model of presheaves as the \textbf{contravariant model}.

Having a strong dialectic between the mathematical framework and physical motivation seems highly desirable, whether we are considering neorealism, we want to gain insight in the foundations of physics,  we are seeking new topos models or seeking connections between the current topos models. 

This concludes the motivation for the research sketched below. The rest of the paper is divided into four parts. Section~\ref{sec: back} covers background information, needed for the later sections. Some constructions from the topos models are covered. In addition there is a brief discussion of topoi as universes of mathematical discourse and locale theory. In Section~\ref{sec: contra} we look at the contravariant model from the internal picture of the presheaf topos at hand. Of particular interest is Subsection~\ref{subsec: physquant}, where we discuss relations between elementary propositions and daseinised self-adjoint operators. In Section~\ref{sec: differences} we look at differences between the contravariant and covariant models on the one hand, and the classical topos of sets. Most of this section is concerned with the physical interpretation of the Heyting algebra structure used in the logic of the topos models. The discussion thus far only considered the kinematics of quantum theory. In Section~\ref{sec: homs}, following ideas of~\cite{nuiten,doering2}, we discuss dynamics. Again, the emphasis is on the internal perspective. As for prerequisites: we assume that the reader is familiar with category theory and von Neumann algebras. Familiarity with topos theory, in particular the internal language of a topos, is highly desirable. 

\section{Background Material} \label{sec: back}

This section is divided into two parts. In the first part we define some key constructions of the contravariant and covariant topos models. In the second subsection we discuss some relevant aspects of topos theory. In particular the internal language, presheaf semantics and locale theory are discussed. The discussion is concise and concentrates on fixing notation and providing references.

\subsection{Background on Topos Models for Quantum Physics} \label{subsec: review}

The discussion below is concise. For more information on the covariant model see~\cite{hls,chls}. For the contravariant model and daseinisation, background information can be found in~\cite{di,di1,di2,di3,di4}. 

The starting point for both the covariant and contravariant model an operator algebra $A$ associated to a physical system. We assume that $A$ is a von Neumann algebra. Typically, the covariant model is applied to a larger class of operator algebras, the unital C*-algebras. However, in this paper we restrict ourselves to von Neumann algebras. An important consequence of this choice is that we can use the daseinisation techniques of the contravariant model also in the covariant case. The covariant model~\cite{hls} also has a C*-algebraic version of daseinisation, but it is unclear to the author how the truth values obtained from this daseinisation map can be interpreted physically.

From the von Neumann algebra $A$, we construct a partially ordered set $\mathcal{C}(A)\equiv\mathcal{C}$. The elements of $\mathcal{C}$ are abelian von Neumann subalgebras of $A$, which we denote by $D,C\in\mathcal{C}$. The partial order is given by inclusion. The elements of $\mathcal{C}$ are called \textbf{contexts} and we think of them as measurement contexts. The idea that $C\in\mathcal{C}$ represents a measurement context is usually avoided in the contravariant model. Most likely, the reason is that people working on this model seek realist interpretations of quantum theory, and want to avoid operationalist notions. 
 
Consider the poset $\mathcal{C}$ as a category. The contravariant approach uses the topos $[\mathcal{C}^{op},\mathbf{Set}]$, the category of presheaves on $\mathcal{C}$. Objects are the contravariant functors $\mathcal{C}^{op}\to\mathbf{Set}$, and arrows are natural transformations between them. Of particular interest is the \textbf{spectral presheaf}, the contravariant functor
\begin{equation} \label{equ: specpresh}
\uS:\mathcal{C}^{\text{op}}\to\mathbf{Set},\ \ \uS(C)=\Sigma_{C},
\end{equation}
where $\Sigma_{C}$ is the Gelfand spectrum of the abelian von Neumann algebra $C\in\mathcal{C}$. Recall that the elements $\lambda\in\Sigma_{C}$ can be identified with multiplicative linear functionals $\lambda:C\to\mathbb{C}$. The operator algebra $C$ is isomorphic to the C*-algebra of continuous complex-valued functions on the compact Hausdorff space $\Sigma_{C}$. If $D\subseteq C$ in $\mathcal{C}$, then the corresponding arrow in the category $\mathcal{C}$ is mapped by $\uS$ to the continuous map $\rho_{CD}:\Sigma_{C}\to\Sigma_{D}$, corresponding by Gelfand duality to the embedding $D\hookrightarrow C$. Note that if we see $\lambda\in\Sigma_{C}$ as a map $C\to\mathbb{C}$, then $\rho_{CD}(\lambda)=\lambda|_{D}$, the restriction of the functional $\lambda$ to $D$.

The covariant approach uses the topos $[\mathcal{C},\mathbf{Set}]$ of covariant functors $\mathcal{C}\to\mathbf{Set}$ and their natural transformations. The key object of this model is the covariant functor
\begin{equation} \label{equ: bohr}
\underline{A}:\mathcal{C}\to\mathbf{Set},\ \ \underline{A}(C)=C.
\end{equation}
If $D\subseteq C$, then the corresponding arrow in $\mathcal{C}$ is mapped by $\underline{A}$ to the inclusion $D\hookrightarrow C$. The object $\underline{A}$ is interesting because, from the internal perspective of the topos $[\mathcal{C},\mathbf{Set}]$ it is a \emph{commutative} unital C*-algebra. There is a version of Gelfand duality which is valid in any (Grothendieck) topos~\cite{banmul3,coq}. Therefore there exists a Gelfand spectrum $\uS_{\underline{A}}$ in $[\mathcal{C},\mathbf{Set}]$ such that $\underline{A}$ is, up to isomorphism of C*-algebras, the algebra of continuous complex-valued functions on $\underline{A}$. However, $\uS_{\underline{A}}$ is not a compact Hausdorff space, but a compact completely regular locale.\\

Next, we give the relevant definitions of daseinisation. For the contravariant model, \textbf{outer daseinisation of projections} is a central construction. Let $p\in\text{Proj}(A)$ be any projection operator of $A$, and $C\in\mathcal{C}$ a context. Then the outer daseinisation of $p$ in $C$ is
\begin{equation}
\delta^{o}(p)_{C}:=\bigwedge\{q\in\text{Proj}(C)\mid q\geq p\}.
\end{equation}
Here we used the fact that $Proj(C)$ is a complete lattice, as this guarantees that the greatest lower bound $\bigwedge$ exists in $C$. The operator $\delta^{o}(p)_{C}$, which is a projection operator in $C$, is the best approximation of $p$ using larger projections. Note that $p\leq\delta^{o}(p)_{C}$ for each $C\in\mathcal{C}$, and $p=\delta^{o}(p)_{C}$ iff $p\in C$. Outer daseinisation is important for the contravariant model as it is used in defining the elementary propositions about the quantum system. Let $a\in A_{sa}$ and $\Delta\subseteq\mathbb{R}$. In orthodox quantum logic the proposition $[a\in\Delta]$ is represented by the spectral projection $\chi_{\Delta}(a)$ of $a$. If $C\in\mathcal{C}$, then $\delta^{o}(\chi_{\Delta}(a))_{C}$ is a projection operator in $C$. By Gelfand duality, a projection operator in $C$ corresponds to a clopen subset $S_{C}$ of $\Sigma_{C}$. In this way, the proposition $[a\in\Delta]$ gives, for each $C\in\mathcal{C}$, a closed open subset of $\Sigma_{C}$. If $D\subseteq C$, then by definition $\delta^{o}(p)_{D}\geq\delta^{o}(p)_{C}$, or, equivalently, $\rho_{CD}(S_{C})\subseteq S_{D}$. This implies that the sets $S_{C}$ combine to a single subobject of the spectral presheaf $\uS$. We will denote this subobject as $\underline{[a\in\Delta]}$.

If $\underline{S}$ is a subobject of the spectral presheaf $\uS$, such that for each $C\in\mathcal{C}$, the set $\underline{S}(C)$ is clopen in $\Sigma_{C}$, then $\underline{S}$ will be called a \textbf{clopen subobject} of $\uS$. In the contravariant model, the Heyting algebra $\mathcal{O}_{cl}\uS$, of clopen subobjects plays an important role. For a discussion of the Heyting algebra structure, see e.g.~\cite[Section 16]{di}. The Heyting algebra structure of $\mathcal{O}_{cl}\uS$ is closely related to the natural Heyting algebra structure of the power object $\mathcal{P}\uS$. From the point of topos theory this seems like a natural choice, but do the operations of the Heyting algebra also make sense physically?

The same question can be asked for the covariant model. In this model elementary propositions are represented by certain opens of the locale $\uS_{\uA}$. By the external description of $\uS_{\uA}$ given in Subsection~\ref{subsub: physquantco}, an internal open of $\uS_{\uA}$ corresponds to giving for each $C\in\mathcal{C}$ an open subset $S_{C}$ of $\Sigma_{C}$, such that ff $D\subseteq C$, then $\rho_{CD}^{-1}(S_{D})\subseteq S_{C}$. Using \textbf{inner daseinisation of projections} we can find such opens. Let, as before, $p\in\text{Proj}(A)$ be a projection operator, and $C\in\mathcal{C}$. Define
\begin{equation}
\delta^{i}(p)_{C}:=\bigvee\{q\in\text{Proj}(C)\mid q\leq p\}.
\end{equation}
If $D\subseteq C$, then $\delta^{i}(p)_{D}\leq\delta^{i}(p)_{C}$, or, equivalently, $\rho_{CD}^{-1}(S_{D})\subseteq S_{C}$. Taking the inner daseinisation of the spectral projection $\chi_{\Delta}(a)$ relative to each context, we can define covariant elementary propositions $\underline{[a\in\Delta]}$ as opens of $\uS_{\underline{A}}$.  However, this is not the way that elementary propositions are usually defined in the covariant model. For the actual elementary propositions, we need daseinisation of self-adjoint operators. 

As a first step, we consider the spectral order on $A_{sa}$, the set of self-adjoint elements of $A$. The spectral order was first considered in~\cite{ols}. Recall that in $A_{sa}$, $a\leq b$, iff there exists an $c\in A$, such that $b-a=c^{\ast}c$.  Also recall that for a von Neumann algebra $A$, to $a\in A_{sa}$ we can associate a family of projections $(e^{a}_{x})_{x\in\mathbb{R}}$, called the spectral resolution of $a$ (see e.g.~\cite{kari}). Define the \textbf{spectral order} $\leq_{s}$, on $A_{sa}$, by
\begin{equation}
a\leq_{s}b\ \ \ \text{iff}\ \ \ \forall x\in\mathbb{R}\ \ e^{a}_{x}\geq e^{b}_{x}.
\end{equation}
If $a\leq_{s}b$, then $a\leq b$, but the converse need not hold, unless $a$ and $b$ are projections, or $a$ and $b$ commute. Using this spectral order, it becomes easy to define the outer and inner daseinisation of self-adjoint operators.
\begin{equation}
\delta^{o}(a)_{C}:=\bigwedge\{b\in C_{sa}\mid b\geq_{s}a\},
\end{equation}
\begin{equation}
\delta^{i}(a)_{C}:=\bigvee\{b\in C_{sa}\mid b\leq_{s}a\}.
\end{equation}
Here we used the fact that $A_{sa}$ is a boundedly complete lattice with respect to the spectral order. Note that we did not need to include the subscript $s$ in the meet or join, as the spectral order coincides with the usual order in the commutative algebra $C$.  

We can view inner and outer daseinisation as (categorical) adjunctions. If we see both $C_{sa}$ and $A_{sa}$ as posets with respect to the spectral order, then the inner and outer daseinisation form right and left adjoints to the inclusion map $i_{C}:C_{sa}\hookrightarrow A_{sa}$. Let $a\in A_{sa}$ and $b\in C_{sa}$. Assume that $b\leq_{s}\delta^{i}(a)_{C}$. As $\delta^{i}(a)_{C}\leq_{s} a$, we conclude that $i_{C}(b)\leq_{s} a$. Conversely, assume that $b\leq_{s} a$. As $\delta^{i}(a)_{C}$ is by definition the join of $b\in C_{sa}$ satisfying $b\leq_{s} a$ we conclude $b\leq_{s}\delta^{i}(a)_{C}$. By an analogous reasoning $a\leq_{s} i(b)$ iff $\delta^{o}(a)_{C}\leq_{s} b$. We conclude that
\begin{equation*}
\delta^{o}(-)_{C}\dashv i_{C}\dashv \delta^{i}(-)_{C}: C_{sa}\to A_{sa}.
\end{equation*}

If $a$ happens to be a projection operator, then the daseinisation as a self-adjoint operator coincides with the daseinisation of $a$ as a projection operator, both in the outer and inner case. If $D\subseteq C$, then
\begin{equation*}
\delta^{i}(a)_{D}\leq_{s}\delta^{i}(a)_{C}\leq_{s}a\leq_{s}\delta^{o}(a)_{C}\leq_{s}\delta^{o}(a)_{D}.
\end{equation*}

For the moment, we reviewed enough constructions from the two topos models. We will explain more as we need them.

\subsection{Background on Topos Theory}

Below, we only treat those aspects of topos theory that are needed for later sections. Readers who know the basics of topos theory should skip this material. Readers having no background in topos theory may find the material below to be hard and sketchy. If this is the case, the reader may want to jump to the next section, or consult one of various excellent texts on topos theory for more information. The standard reference is the book \textit{Sheaves in Geometry and Logic} by Mac Lane and Moerdijk~\cite{mm}.  An easier, but less in-depth introduction is~\cite{gol}. The books~\cite{bor, bell} provide a useful resource about the internal language of a topos. Of course, everything and more about topos theory can be found in the massive~\cite{jh1}. The texts~\cite{jh4,pipu} provide more information about locales.

\subsubsection{Topoi as Generalised Spaces} \label{subsub: spaces}

There are various ways of looking at topoi, and it is the interplay between these different points of view that makes topos theory interesting. For our goals, thinking of a topos as a generalised universe of sets, or a universe of mathematical discourse, is the most important perspective. In this subsection however, we think of topoi as generalisations of topological spaces. This helps in introducing geometric morphisms and locales, two important notions in later sections.

If $X$ is a topological space, then we can associate to it the topos $Sh(X)$, of sheaves on that topological space~\cite[Chapter II]{mm}. From the topos $Sh(X)$ we can recover the topology $\mathcal{O}X$ of $X$ by considering the subobject classifier $\Omega$. If the space $X$ is Hausdorff (and hence sober),  we can subsequently retrieve $X$ as the set of points of $\mathcal{O}X$. 

Given $f:X\to Y$, a continuous map of topological spaces, we can associate to $f$ a geometric morphism $F:Sh(X)\to Sh(Y)$. For topoi $\mathcal{E}$, and $\mathcal{F}$, a \textbf{geometric morphism} $F:\mathcal{E}\to\mathcal{F}$ is an adjoint pair of functors where the right adjoint $F_{\ast}:\mathcal{E}\to\mathcal{F}$ is called the \textbf{direct image functor}, and the left adjoint $F^{\ast}:\mathcal{F}\to\mathcal{F}$ is called the \textbf{inverse image functor} and where $F^{\ast}$ is assumed to be left-exact~\cite[Chapter VII]{mm}. Conversely, if $F: Sh(X)\to Sh(Y)$ is a geometric morphism and the spaces are Hausdorff, then $F$ comes from a unique continuous map $f:X\to Y$.  

The topos $Sh(X)$ depends on the topology $\mathcal{O}X$, rather than on the underlying set of points of the space. For example, if $X$ has the trivial topology $\{\emptyset,X\}$, then $Sh(X)$ can be identified with $\mathbf{Set}$, regardless of the set $X$. In this sense we should not see a topos as a generalisation of topological spaces, but of locales. In order to define locales, we first need to consider frames. A \textbf{frame} $F$ is a complete lattice, satisfying the following distributivity law
\begin{equation*}
\forall U\in F,\ \ \forall S\subseteq F,\ \  U\wedge\left(\bigvee S\right)=\bigvee\{U\wedge V\mid V\in S\}.
\end{equation*}
The motivating example of a frame is a topology $\mathcal{O}X$, where $\wedge$ is simply the intersection $\cap$, and $\bigvee$ corresponds to the union $\bigcup$. However, not every frame comes from a topology, as we will see below. If $F$ and $G$ are frames, then a morphism of frames, or a \textbf{frame homomorphism} $f:F\to G$, is a function that preserves finite meets and all joins. If $f:X\to Y$ is a continuous map of topological spaces, then the inverse image map $f^{-1}:\mathcal{O}Y\to\mathcal{O}X$ is a frame homomorphism. 

The category of locales $\mathbf{Loc}$ is defined to be the dual category of the category of frames $\mathbf{Frm}$. So a \textbf{locale} $L$ corresponds to a unique frame, which we denote as $\mathcal{O}L$, and a \textbf{locale map}, $f:K\to L$, also called a continuous map, is the same as a frame map $f^{\ast}:\mathcal{O}L\to\mathcal{O}K$. In particular, a topological space $X$ defines a locale $L(X)$ through the topology $\mathcal{O}X$, and a continuous map of spaces $f:X\to Y$ induces, through the inverse image map, a locale map $L(f):L(X)\to L(Y)$.

Points of a space $X$ correspond to continuous maps $x:1\to X$, where $1$ is the one-point topological space. The inverse image of $x:1\to X$ is a frame morphism $x^{-1}:\mathcal{O}X\to\underline{2}$, where $\underline{2}$ is the frame of two elements. If we are given the topology $\mathcal{O}X$, we consider the set $pt(X)$ consisting of frame homomorphisms $\mathcal{O}X\to\underline{2}$. The frame $\mathcal{O}X$ defines a topology on $pt(X)$: if $U\in\mathcal{O}X$, let $pt(U)$ consist of the $p\in pt(X)$ such that $p(U)$ is the top element of $\underline{2}$. The space $X$ is called \textbf{sober} iff it is homeomorphic to $pt(X)$. In other words, we can reconstruct the points from the topology. In particular, any Hausdorff space is sober. For any sober space $X$, we can retrieve $X$ from the locale $L(X)$, or from the topos $Sh(X)$.

As mentioned before, a frame, and therefore a locale, need not come from a topological space. For example, let $F$ consist of the subsets $U$ of the real line $\mathbb{R}$, satisfying the condition that taking the interior of the closure of $U$ is equal to $U$. This set, counting all open intervals $(r,s)$ amongst its elements, and partially ordered by inclusion, defines a frame. However, contrary to topologies, there are no points in the sense of frame homomorphisms $F\to 2$. The only topology that has no points is the unique topology on $\emptyset$, and obviously $F$ is not isomorphic to this topology.

The definition of frames and their morphisms can be interpreted in the internal language of any topos $\mathcal{E}$. As a consequence we can work with frames and locales internal to a topos. We proceed to discussing the internal language of topoi.

\subsubsection{Topoi as Generalised Universes of Sets} \label{subsub: sets}

Whenever we talk about taking an \textbf{internal} perspective or internal picture of the topos $[\mathcal{C}^{op},\mathbf{Set}]$, this means that we are working with the objects and arrows of that topos without referring to the topos $\mathbf{Set}$. This could mean that we are considering these objects and arrows in terms of abstract category theory. However, in this paper it means that we are using the \textbf{internal language} of the topos. Any topos has an associated internal or Mitchell-B\'enabou language~\cite{bor,bell2,mm}. With respect to this language, working with objects and arrows of the topos resembles set theory, but the category $\mathbf{Set}$ itself is not used. Whenever we use $\mathbf{Set}$ in our descriptions, we take on an \textbf{external} perspective. For example, if we think of the spectral presheaf $\uS$ as a functor mapping into the category $\mathbf{Set}$, we are using an external perspective. If we consider $\uS$ as a `set' (or topological space as we shall see in the next section) with respect to the internal language, we are dealing with an internal perspective. 

The internal language of a topos $\mathcal{E}$ uses the objects of that topos as types, and arrows as terms. The language generated by these terms is rich enough to express complicated mathematics. One type (i.e., object) is of particular importance, and is denoted by $\Omega$. For a topos $\mathcal{E}$, the \textbf{subobject classifier} is an object $\Omega$, together with an arrow $\mathbf{true}:1\to\Omega$, such that for any subobject $U\hookrightarrow X$, there is a unique arrow $X\to\Omega$ making the following square into a pullback.
\[ \xymatrix{ U \ar[r]^{!} \ar@{^{(}->}[d] & 1 \ar[d]^{\mathbf{true}}\\
X \ar[r] & \Omega}\]
For the topos $\mathbf{Set}$, $\Omega$ is the two-element set $\underline{2}$, and the function corresponding to an inclusion $U\hookrightarrow X$ is just the characteristic function $\chi_{U}:X\to\underline{2}$. Returning to an arbitrary topos, the arrows $1\to\Omega$ are called the \textbf{truth values}. Like the topos $\mathbf{Set}$ there are always truth values $\mathbf{true}$ and $\mathbf{false}$, but unlike $\mathbf{Set}$, there may be many other truth values. For the topos $[\mathcal{C}^{op},\mathbf{Set}]$, the truth values correspond bijectively to the sieves on $\mathcal{C}$, and for $[\mathcal{C},\mathbf{Set}]$ the truth values can be identified with the cosieves on $\mathcal{C}$.

We return to the internal language of a topos. Let $\phi(x,y)$ be a formula in this language, with some free variable $x$ of type $X$ (free means that it is not bound by a quantifier), and a free variable $y$ of type $Y$, then this formula has an interpretation as a subobject
\begin{equation*}
\{(x,y)\in X\times Y\mid \phi(x,y)\}\subseteq X\times Y.
\end{equation*}
By definition of the subobject classifier, this corresponds to a unique arrow $X\times Y\to\Omega$. If $\psi$ is a proposition, a formula without free variables, then it corresponds to a subobject of the terminal object $1$, or, equivalently, a truth value. We say that $\psi$ \emph{is true} and write $\Vdash\psi$, if $\psi$ corresponds to the truth value $\mathbf{true}$.

As noted before, with respect to the internal language the topos becomes a universe of mathematical discourse, resembling set theory. However, there are some important differences. In the topos $\mathbf{Set}$, an element of a set $X$ corresponds to a function $x:1\to X$, where $1$ is a singleton set, the terminal object of $\mathbf{Set}$. For an arbitrary topos $\mathcal{E}$ we need a more generalised notion of element. An element of a `set' (i.e., object) $X$ is any arrow $Y\to X$.

In the topos $\mathbf{Set}$ these generalised elements of a set $X$ are \textbf{all} functions $f:Y\to X$, that have $X$ as a codomain. We know from set theory that we need not bother with $Y\neq1$. This follows from the observation that the object $1$ generates the category $\mathbf{Set}$. This means that for any pair of functions $f,g:X\to Y$, the functions are equal $f=g$, iff $f\circ x=g\circ x$ for each $x:1\to X$.

In fact, for any Grothendieck topos there is a set of objects which generate the topos, and can therefore be used to reduce the number of generalised elements that we need to consider. If $y:\mathcal{C}\to[\mathcal{C}^{op},\mathbf{Set}]$ denotes the Yoneda embedding, then the objects $y(C)$ generate the presheaf category $[\mathcal{C}^{op},\mathbf{Set}]$, meaning that for any pair of natural transformations $f,g:X\to Y$, we have $f=g$ iff for each element of the form $x:y(C)\to X$, $f\circ x=g\circ x$. This means that in doing internal mathematics, we can restrict to generalised elements of the form $y(C)\to X$. Exploiting this observation leads to \textbf{presheaf semantics} which we use several times throughout the text when proving internal claims.

By the Yoneda Lemma, a generalised element $\alpha:y(C)\to X$ corresponds to an element $\alpha\in X(C)$. The following conditions are equivalent:
\begin{enumerate}
\item The generalised element $\alpha$ factors through the inclusion $\{x\mid\phi(x)\}\subseteq X$;
\item $\alpha\in\{x\mid\phi(x)\}(C)$. 
\end{enumerate}
We denote these two equivalent conditions by the \textbf{forcing relation} $C\Vdash\phi(\alpha)$. Note that we can think of $C\Vdash\phi(\alpha)$, either internally (no reference to $\mathbf{Set}$) using condition (1) expressing that the generalised element $\alpha$ of the `set' $X$ is an element of the `set'  $\{x\in X\mid\phi(x)\}$, or, externally using condition (2). The internal and external view are connected by the Yoneda Lemma.\\

For the copresheaf topos $[\mathcal{C},\mathbf{Set}]$, we can also use a version of presheaf semantics. This amounts to replacing the contravariant Hom-functor $y$ of the Yoneda embedding by the covariant Hom-functor $k$, and considering the generating set of objects $k(C)$. Presheaf semantics, its rules, and the more general sheaf semantics are explained in~\cite[Section VI.7]{mm}. 

\section{The Contravariant Model, seen Internally} \label{sec: contra}

Key concepts of the contravariant model, such as the spectral presheaf, possible value objects, elementary propositions, daseinised self-adjoint operators, and states, can be viewed from the internal perspective of the topos at hand. From this internal perspective these objects and arrows resemble constructions from classical physics more closely. In this section, we also investigate connections between elementary propositions $[a\in\Delta]$ and daseinised self-adjoint operators $\underline{\delta(a)}$. Such a connection need not be simple, as the elementary propositions are defined using subsets $\Delta$ of the (external) real numbers $\mathbb{R}$, whereas $\underline{\delta(a)}$ takes values in an (internal) value object $\mathcal{R}$ of the topos. Establishing such a connection is therefore about relating the real numbers of the topos $\mathbf{Set}$ to the value object $\mathcal{R}$ of the topos $[\mathcal{C}^{op},\mathbf{Set}]$. 

\subsection{Spectral Presheaf as a Space} \label{subsec: space}

We start with the central object of the contravariant model, the spectral presheaf $\uS$, which we defined in Subsection~\ref{subsec: review}. In the contravariant model, the object $\uS$ of the topos $\topos$ is thought of as a state space, in analogy with classical physics. But is there any mathematical justification for calling $\uS$ a space? We can think of $\uS$ either as an object in an abstract category, or as a functor taking values in the category of topological spaces and continuous maps. In the internal language of $[\mathcal{C}^{op},\mathbf{Set}]$, $\uS$ is just a set, and we can always consider a set as a discrete space. However, we can do better than that. Below we describe $\uS$ as a topological space internal to the topos $\topos$, in such a way that states on $A$ (in the sense of normalised positive functionals) and (daseinised) self-adjoint operators have a clear internal perspective. This internal perspective strengthens the analogy with classical physics, as well as with the copresheaf model.

Given a set $X$  (in the internal sense) in a topos $\mathcal{E}$, a topology $\mathcal{O}X$ on $X$ is defined in a straightforward way. It is a subset $\mathcal{O}X\subseteq\mathcal{P}X$ of the powerset of $X$, satisfying the condition
\begin{align*}
\Vdash (X\in\mathcal{O}X) &\wedge(\emptyset\in\mathcal{O}X)\wedge(U,V\in\mathcal{O}X\rightarrow U\cap V\in\mathcal{O}X)\\
& \wedge(Y\subseteq\mathcal{O}X\rightarrow\bigcup Y\in\mathcal{O}X).
\end{align*}

For a topos of the kind $\mathcal{E}=Sh(T)$ with $T$ a topological space, there is a useful external description of internal topologies on a sheaf $X$, as explained in~\cite{moerdijk}. This is relevant because the topos $\topos$ is equivalent to the topos $Sh(\mathcal{C}_{\downarrow})$, where the space $\mathcal{C}_{\downarrow}$ is the set $\mathcal{C}$ equipped with the downwards Alexandroff topology. With respect to this topology $U\subseteq\mathcal{C}$ is open iff it is downwards closed in the sense that if $C\subseteq C'\in U$, then $C\in U$. We write $\mathcal{C}^{\downarrow}_{A}$ for $\mathcal{C}_{\downarrow}$ if we want to emphasise that $\mathcal{C}$ comes from $A$. This will become important in Section~\ref{sec: homs}, where we consider $\ast$-homomorphisms.

For the external description of internal topologies (i.e., a description in $\mathbf{Set}$), first recall that the category $Sh(T)$ is equivalent to the category $\textbf{\text{\'Etale}}(T)$, of \'etale bundles over $T$ (as explained in detail in~\cite[Chapter II]{mm}). An \'etale bundle over $T$ is a continuous map $p:X\to T$ such that each $x\in X$ has an open neighbourhood $U_{x}$, satisfying the condition that $p|_{U_{x}}:U_{x}\to T$ is a homeomorphism onto its image.

Under the identifications $\topos\cong Sh(\mathcal{C}_{\downarrow})\cong \textbf{\text{\'Etale}}(\mathcal{C}_{\downarrow})$, the spectral presheaf $\uS$ corresponds to the \'etale bundle $\pi:\Sigma_{e}\to\mathcal{C}_{\downarrow}$, where the set $\Sigma_{e}$ is the disjoint union of Gelfand spectra $\coprod_{C\in\mathcal{C}}\Sigma_{C}$. This set is equipped with the following (\'etale) topology. For any non-empty $U\subseteq\Sigma_{e}$, we have $U\in\mathcal{O}\Sigma_{e}$ iff the following condition holds. If $(C,\lambda)\in U$ (with $\lambda\in\Sigma_{C}$), and $D\subseteq C$, then $(D,\lambda|_{D})\in U$. Here we used the notation $\lambda|_{D}$ for the restriction $\rho_{CD}(\lambda)$. The function $\pi$ is simply given by the projection $(C,\lambda)\mapsto C$. 

\begin{poe}{(\cite{moerdijk})}
There is a bijection between topologies on $\uS$, internal to $\topos$, and $\pi$-topologies on $\Sigma_{e}$. A topology on $\Sigma_{e}$ is called a $\pi$-topology if it is coarser than the \'etale topology and with respect to which $\pi$ is continuous.
\end{poe}

Note that the \'etale topology itself qualifies, and this corresponds to the discrete topology on $\uS$. It is not hard to see that the \'etale opens of $\Sigma$ correspond to subobjects of $\uS$. In the contravariant topos approach one is typically only interested in subobjects of $\uS$ of a certain kind, the clopen subobjects. Recall that a subobject of $\underline{U}\subseteq\uS$ is a clopen subobject iff for each $C\in\mathcal{C}$ the subset $\underline{U}(C)\subseteq\Sigma_{C}$ is clopen with respect to the topology on the Gelfand spectrum $\Sigma_{C}$. In the external description $\pi:\Sigma_{e}\to\mathcal{C}_{\downarrow}$, the clopen subobjects correspond to the \'etale opens $U$ of $\Sigma_{e}$ satisfying the condition that for each $C\in\mathcal{C}$, the set $U_{C}:=U\cap\Sigma_{C}$ is clopen in $\Sigma_{C}$. These \'etale opens are not closed under infinite unions and therefore do not form a topology. However, they do form a basis for a topology. Note that, since we are working with von Neumann algebras, each $\Sigma_{C}$ has a basis of clopen subsets. By this observation, the internal topology of $\uS$ generated by the clopen subobjects can be presented externally as follows.

\begin{dork}
The space $\Sigma_{\downarrow}$ is the set $\Sigma_{e}$, where $U\subseteq\Sigma_{e}$ is open iff the following two conditions are satisfied
\begin{enumerate}
\item If $\lambda\in U_{C}$ and $C'\subseteq C$, then $\lambda|_{C'}\in U_{C'}$.
\item For every $C\in\mathcal{C}$, $U_{C}$ is open in $\Sigma_{C}$
\end{enumerate}
\end{dork}

\begin{poe} \label{prop: bundle}
The topology $\mathcal{O}\Sigma_{\downarrow}$ is a $\pi$-topology, and thus defines an internal topology on $\uS$, which is the topology generated by the closed open subobjects.
\end{poe}

In what follows, we write $\Sigma_{A}^{\downarrow}$ for $\Sigma_{\downarrow}$ whenever we want to emphasise $A$. We write $\uS_{\downarrow}$ for $\uS$ with the internal topology generated by the clopen subobjects.

\begin{rem}
As shown in~\cite[Section C1.6]{jh1}, the category $\mathbf{Loc}_{Sh(T)}$, of locales in $Sh(T)$, is equivalent to the category $\mathbf{Loc}/T$, of locales (in $\mathbf{Set}$) over $T$. In particular, a topological space internal to $Sh(T)$ corresponds to a locale in $Sh(T)$ as such a space is described by a bundle over $T$. Externally, passing from topological spaces to locales means that we forget that we are working with topologies that are coarser than the \'etale topology of some sheaf. Internally, passing from topological spaces to locales means that we forget about the set of points that we topologised. The bundle $\pi:\Sigma_{\downarrow}\to\mathcal{C}_{\downarrow}$, where $\Sigma_{\downarrow}$ has the topology of Proposition~\ref{prop: bundle}, perceived internally as a locale rather than as a topological space, was earlier considered in~\cite[Subsection 2.1]{wollie}.
\end{rem}

\subsubsection{Sobriety}

If we think of the spectral presheaf of the contravariant model as an internal space or locale, we can ask if it satisfies any separation properties. As shown in~\cite[Section 2.1]{wollie}, the spectral presheaf, seen as an internal locale, is in general not regular, and is therefore not the spectrum of some internal unital commutative C*-algebra. However, we prove the following weaker property.

\begin{lem}
The internal space $\underline{\Sigma}_{\downarrow}$, associated to the bundle $\pi_{\downarrow}:\Sigma_{\downarrow}\to\mathcal{C}_{\downarrow}$ of Proposition~\ref{prop: bundle} is sober (internally).
\end{lem}

\begin{proof}
First assume that $\mathcal{C}_{\downarrow}$ is a sober space, such as for $A=M_{n}(\mathbb{C})$. Let $\pi_{e}:\Sigma_{e}\to\mathcal{C}_{\downarrow}$, and $\pi_{\downarrow}:\Sigma_{\downarrow}\to\mathcal{C}_{\downarrow}$ denote the bundles associated to $\uS$ and $\uS_{\downarrow}$. The space $\Sigma_{\downarrow}$ and the \'etale space $\Sigma_{e}$ are both sober because $\mathcal{C}_{\downarrow}$ is sober. Let $j:\Sigma_{e}\to\Sigma_{\downarrow}$ be the function $(C,\lambda)\mapsto(C,\lambda)$, corresponding to the inclusion $j^{-1}:\mathcal{O}\Sigma_{\downarrow}\hookrightarrow\mathcal{O}\Sigma_{e}$. By~\cite[Corollary 3.2]{moerdijk}, the space $\underline{\Sigma}_{\downarrow}$ is sober internally iff the function $s\mapsto s\circ j$, mapping continuous sections of the bundle $\pi_{e}$ to continuous sections of the bundle $\pi_{\downarrow}$ is a bijection.

Let $U\in\mathcal{O}\mathcal{C}_{\downarrow}$, and $s:U\to\Sigma_{\downarrow}$ be a continuous section of $\pi_{\downarrow}$. For convenience, we write $s(C)=(C,r(C))$, with $r(C)\in\Sigma_{C}$. We will show that if $D\subseteq C$, then $r(D)=r(C)|_{D}$. When we know this, we can conclude that $s$ is also a continuous section of the bundle $\pi_{e}$. Let $\lambda\in\Sigma_{D}$ be different from $r(C)|_{D}$. Choose an open neighbourhood $U$ of $r(C)|_{D}$ in $\Sigma_{D}$ such that $\lambda\notin U$. Next, consider $V_{C}=\rho^{-1}_{CD}(U)$, an open neighbourhood of $r(C)$ in $\Sigma_{C}$. Define the open $V$ in $\Sigma_{\downarrow}$ as follows. If $D'\subseteq C$, then $V_{D'}=\rho_{CD'}(V_{C})$. If $D'$ is not below $C$, then  $V_{D'}=\emptyset$. Note that we used the fact that the restriction maps $\rho_{CD}$ are open maps, in order for $V$ to be open in $\Sigma_{\downarrow}$. By construction $C\in s^{-1}(V)$. If follows that  $D\in s^{-1}(V)$, as $s$ is continuous. We conclude that $\lambda\neq r(D)$. As $\lambda$ was an arbitrary element of $\Sigma_{D}$ such that $\lambda\neq r(C)|_{D}$, we conclude  $r(D)=r(C)|_{D}$. We can conclude that $j$ induces an isomorphism on the sections, and consequently, that $\underline{\Sigma}_{\downarrow}$ is sober internally. The completes the proof for the cases where $\mathcal{C}_{\downarrow}$ is sober. Next, drop this assumption.

The sobrification of $\mathcal{C}_{\downarrow}$ can be identified with the set $\mathcal{F}$ of filters of $\mathcal{C}$, equipped with the Scott topology. Recall that a subset $F\subseteq\mathcal{C}$ is a filter iff it is non-empty, upward closed, and downward directed. A subset $W\subseteq\mathcal{F}$ is Scott open iff it is upward closed with respect to the inclusion relation, and if for any directed family of filters $(F_{i})_{i\in I}$, satisfying $\bigcup_{i\in I}F_{i}\in W$, implies that there exists an $i_{0}\in I$, such that  $F_{i_{0}}\in W$. The Scott topology on $\mathcal{F}$ is generated by the basis
\begin{equation*}
W_{C}=\{F\in\mathcal{F}\mid C\in F\},\ \ \ C\in\mathcal{C}.
\end{equation*}
The continuous map 
\begin{equation*}
i:\mathcal{C}_{\downarrow}\to\mathcal{F},\ \ i(C)=\uparrow C:=\{E\in\mathcal{C}\mid E\supseteq C\}, 
\end{equation*}
defines, through its inverse image, an isomorphism of frames
\begin{equation*}
i^{-1}:\mathcal{O}\mathcal{F}\to\mathcal{O}\mathcal{C},\ \ i^{-1}(W_{C})=\downarrow C.
\end{equation*}
Using this frame isomorphism we identify $Sh(\mathcal{C}_{\downarrow})$ with $Sh(\mathcal{F})$. Let $\pi_{\mathcal{F}}:\Sigma_{\mathcal{F}}\to\mathcal{F}$ be the \'etale bundle corresponding to the spectral presheaf. Using the observation that for the principal filter $(\uparrow C)$, the smallest Scott open neighborhood is $W_{C}$, we identify the fibre $\pi^{-1}_{\mathcal{F}}(\uparrow C)$ with $\Sigma_{C}$. If we see $\lambda\in\Sigma_{C}$ as an element of $\lambda\in\uS(W_{C})$, and $F$ is any filter in $\mathcal{C}$ containing $C$, let $[\lambda]_{F}$ denote the germ of $\lambda$ in $F$. Note that the (\'etale) topology on $\Sigma_{\mathcal{F}}$ is generated by the basis
\begin{equation*}
B_{C,\lambda}=\{(F,[\lambda]_{F})\in\Sigma_{\mathcal{F}}\mid F\in W_{C}\}\ \ C\in\mathcal{C},\ \ \lambda\in\Sigma_{C},
\end{equation*}
whereas the topology on $\Sigma^{\downarrow}_{\mathcal{F}}$ is generated by the coarser basis
\begin{equation*}
B_{C,u}=\{(F,[\lambda]_{F})\in\Sigma_{\mathcal{F}}\mid\ F\in W_{C},\ \ \lambda\in u\},\ \ C\in\mathcal{C},\ \ u\in\mathcal{O}\Sigma_{C}.
\end{equation*}
Using the same reasoning as for sober $\mathcal{C}_{\downarrow}$, any continuous section of the bundle $\Sigma^{\downarrow}_{\mathcal{F}}\to\mathcal{F}$ is a continuous section of $\Sigma_{\mathcal{F}}\to\mathcal{F}$, which is enough to conclude that the space $\underline{\Sigma}_{\downarrow}$ is sober internally, even if $\mathcal{C}_{\downarrow}$ is not sober externally.
\end{proof}

\subsubsection{Continuous Maps}

Evidently, if we want to talk about topological spaces in topoi, we also want to talk about continuous maps. Let $(\underline{X},\mathcal{O}\underline{X})$ and $(\underline{Y},\mathcal{O}\underline{Y})$ be two topological spaces in $Sh(T)$ externally described by bundles $p:X\to T$, and $q:Y\to T$ respectively. A continuous map $f:(\underline{X},\mathcal{O}\underline{X})\to(\underline{Y},\mathcal{O}\underline{Y})$ is a sheaf morphism $f:\underline{X}\to\underline{Y}$ satisfying
\begin{equation*}
\Vdash \forall U\in\mathcal{P}\underline{Y}\ \ (U\in\mathcal{O}\underline{Y})\Rightarrow(f^{-1}(U)\in\mathcal{O}\underline{X}).
\end{equation*}
The sheaf morphism $f^{-1}:\mathcal{P}\underline{Y}\to\mathcal{P}\underline{X}$ used in this condition is described in \cite[SectionIV.1]{mm}, where it is aptly called $\mathcal{P}f$. Under the identification of $Sh(T)$ with $\textbf{\text{\'Etale}}(T)$, sheaf morphisms $\underline{f}:\underline{X}\to\underline{Y}$ correspond to commuting triangles of continuous functions
\[ \xymatrix{ X \ar[rr]^{f} \ar[dr]_{p} & & Y \ar[dl]^{q}\\
& T &} \] 
Here the map $f$ is continuous with respect to the \'etale topologies on $X$ and $Y$. Such a map $f$ corresponds to an internal continuous map $\underline{f}:(\underline{X},\mathcal{O}\underline{X})\to(\underline{Y},\mathcal{O}\underline{Y})$ iff, in addition, $f$ is continuous with respect to the coarser topologies on $X$ and $Y$ coming from the internal topologies $\mathcal{O}\underline{X}$. and $\mathcal{O}\underline{Y}$. We return to internal continuous maps after discussing the value objects of the contravariant topos approach.

\subsection{Spaces of values}

In the previous subsection the state space object of the contravariant approach, the spectral presheaf $\uS$, has been given the structure of an internal topological space. In this subsection we concentrate on the value object of the contravariant approach. The value object is thought of as the space of values for physical quantities. This object of values need not be the real numbers (insofar as one can even speak of \textbf{the} real numbers in a topos). As sketched in e.g.~\cite{ish}, one of the aims of these topos models is to investigate alternative spaces of values. This is because relying on real numbers may turn out to be problematic for theories of quantum gravity. In this subsection, we see how the value object of the contravariant approach is related to internal real numbers.

Let $\mathbb{R}$ denote the real numbers in the topos $\mathbf{Set}$, and let $P$ be a poset. In what follows $\text{OP}(P,\mathbb{R})$ denotes the set of order-preserving functions $r:P\to\mathbb{R}$ and $\text{OR}(P,\mathbb{R})$ denotes the set of order-reversing functions $s:P\to\mathbb{R}$. We write $r\leq s$ if $r(p)\leq s(p)$ for all $p\in P$. A popular choice for the value object in the contravariant approach is the functor $\underline{\mathbb{R}}^{\leftrightarrow}:\mathcal{C}^{op}\to\mathbf{Set}$, defined by
\begin{equation} \label{equ: valobj}
\underline{\mathbb{R}}^{\leftrightarrow}(C)=\{(r,s)\in\text{OP}(\downarrow C,\mathbb{R})\times\text{OR}(\downarrow C,\mathbb{R})\mid r\leq s\},
\end{equation}
where the restriction map corresponding to the inclusion $D\subseteq C$, maps $(r,s)$ to $(r|_{\downarrow D},s|_{\downarrow D})$. This object is closely related to two different kinds of real numbers in the topos $[\mathcal{C}^{op},\mathbf{Set}]$. Using the natural numbers $\underline{\mathbb{N}}$ of this topos, we can construct real numbers as we would in the topos $\mathbf{Set}$. However, the axiom of choice and law of excluded middle are not validated in the presheaf topos $[\mathcal{C}^{op},\mathbf{Set}]$. This entails that constructions that yield the same set of real numbers in the topos $\mathbf{Set}$, may yield different objects in the topos $[\mathcal{C}^{op},\mathbf{Set}]$. In particular, we will be interested in the three versions of real numbers in the following definition.

\begin{dork}
Consider the following versions of real numbers:
\begin{itemize}
\item The \textbf{lower real numbers}, $\mathbb{R}_{l}$, are the rounded down-closed subsets of $\mathbb{Q}$, where $\underline{x}\subseteq\mathbb{Q}$ is called rounded if $p\in\underline{x}$ implies that there exists a $p<q\in\mathbb{Q}$ such that $q\in\underline{x}$, and $\underline{x}\subseteq\mathbb{Q}$ is called down-closed if $p<q\in\underline{x}$ implies that $p\in\underline{x}$. If $\underline{x}\in\mathbb{R}_{l}$ and $q\in\mathbb{Q}$, then we write $q<\underline{x}$ whenever $q$ is in $\underline{x}$.
\item The \textbf{upper real numbers}, $\mathbb{R}_{u}$ , are the rounded up-closed subsets of $\mathbb{Q}$. In this case rounded means that if $p\in\bar{x}$ then there exists a $q<p$ such that $q\in\bar{x}$.
If $\bar{x}\in\mathbb{R}_{u}$ and $q\in\mathbb{Q}$, then we write $\bar{x}<q$ whenever $q$ is in $\bar{x}$.
\item The \textbf{Dedekind real numbers}, $\mathbb{R}_{d}$, are pairs $\langle\underline{x},\bar{x}\rangle$, where $\underline{x}\in\mathbb{R}_{l}$ is non-empty, $\bar{x}\in\mathbb{R}_{u}$ is non-empty, $\underline{x}\cap\bar{x}=\emptyset$, and $\underline{x}$ and $\bar{x}$ are arbitrarily close, in that if $q,r\in\mathbb{Q}$, with $q<r$, then either $q<\underline{x}$ or $\bar{x}>r$.
\end{itemize}
\end{dork}

Note that by the above definition the sets $\mathbb{Q}$ and $\emptyset$ are lower and upper real numbers. If we exclude $\mathbb{Q}$ and $\emptyset$, we note that in the topos $\mathbf{Set}$, any lower real can be identified with its supremum, and any upper real with its infimum. Therefore all three versions of real numbers given above can be identified with each other and with $\mathbb{R}$ (or with $\mathbb{R}$ extended with $\{-\infty,+\infty\}$, if we want to include $\mathbb{Q}$ and $\emptyset$). The definitions given above make sense internally to every topos that has a natural numbers object, and hence in particular to every Grothendieck topos. In such a topos $\mathcal{E}$, these constructively different notions of real numbers need not correspond to the same object, as they do in $\mathbf{Set}$.

Even though the sets $\mathbb{R}_{l}$, $\mathbb{R}_{u}$ and $\mathbb{R}_{d}$ coincide in $\mathbf{Set}$, the natural topologies on these sets differ. The topology on $\mathbb{R}_{l}$ is the topology generated by upper half intervals $(y,+\infty]$, $y\in\mathbb{R}$. The topology on $\mathbb{R}_{u}$ is the topology generated by half open intervals $[-\infty,y)$, with $y\in\mathbb{R}$. The topology on $\mathbb{R}_{d}$ is the familiar Hausdorff topology on $\mathbb{R}$ generated by the open intervals $(x,y)$, with $x,y\in\mathbb{R}$.

The previous statement requires some clarification. In what sense are these topologies natural? Each of the real numbers of the definition can be captured by a propositional theory, within the constraints of geometric logic~\cite{vic, jh1}. To such a theory we can associate a frame, just like one associates a Lindenbaum algebra to a classical propositional theory. The points of this frame are the (standard) models of the theory, which, in our case, are the real numbers. The topologies that we consider are the (Lindenbaum) frames of the corresponding theories.

What do the lower, upper and Dedekind reals look like in the topos $[\mathcal{C}^{op},\mathbf{Set}]$? As these reals are defined by a geometric propositional theory, we can view them as either locales (whose frame is constructed like the Lindenbaum algebra of a classical propositional theory) or as sets (the set of models of the theory). We will also describe them as internal topological spaces, which will be convenient when we consider daseinised self-adjoint operators.

The fact that these real numbers are described by propositional geometric theories also makes it easy to find the external description of their frames. Under the identification of the category of locales internal to $[\mathcal{C}^{op},\mathbf{Set}]$ with the category of locales over $\mathcal{C}_{\downarrow}$ with the Alexandroff down set topology, the different kinds of real numbers are given by the following bundles.

\begin{lem}
The external description of the locales of lower, upper and Dedekind real numbers in $[\mathcal{C}^{op},\mathbf{Set}]$ is given by the bundles
\begin{equation} \label{equ: lower}
\pi_{1}:\mathcal{C}_{\downarrow}\times\mathbb{R}_{\alpha}\to\mathcal{C}_{\downarrow},\ \ (C,x)\mapsto C,
\end{equation}
where for $\alpha$ we may take $l$, $u$ or $d$, and $\mathbb{R}_{\alpha}$ is viewed as a topological space in $\mathbf{Set}$ with the topologies given above. 
\end{lem}

A discussion why (\ref{equ: lower}) gives the right description can be found in\footnote{Actually, in \cite{jh1} it is assumed that we are working over a sober space. We could consider the sobrification of $\mathcal{C}^{\downarrow}_{A}$ and consider the upper, lower and Dedekind reals over this space. However, this leads to the same frames we are using.} of \cite[Section D4.7]{jh1}. The bundle (\ref{equ: lower}), with $\alpha=l$ describes the lower reals as a locale. The corresponding internal set of lower reals in $[\mathcal{C}^{op},\mathbf{Set}]$ (the set of points of the locale) is given by the functor
\begin{equation*}
\underline{\mathbb{R}}_{l}:\mathcal{C}^{op}\to\mathbf{Set}, \ \ \underline{\mathbb{R}}_{l}(C)=C((\downarrow C),\mathbb{R}_{l}),
\end{equation*}
the presheaf of (Alexandroff) continuous functions taking values in $\mathbb{R}_{l}$. For any topological space $X$, a function $\mu:X\to\mathbb{R}_{l}$ is continuous iff it is lower semicontinuous, when seen as a function $\mu:X\to\mathbb{R}$. By definition of the down set topology on $\mathcal{C}$ the function $\mu$ is lower semicontinuous iff it is order reversing. In the contravariant approach the following presheaf plays an important r\^ole, see for example~\cite[Definition 8.2]{di}:
\begin{equation*}
\underline{\mathbb{R}}^{\preceq}:\mathcal{C}^{op}\to\mathbf{Set},\ \ \mathbb{R}^{\preceq}(C)=\text{OR}((\downarrow C),\mathbb{R}).
\end{equation*}
We recognise the presheaf $\underline{\mathbb{R}}^{\preceq}$ as the presheaf of lower real numbers $\underline{\mathbb{R}}_{l}$. In the same way the sets of upper and Dedekind real numbers in $[\mathcal{C}^{op},\mathbf{Set}]$ can be described.

\begin{lem}
Externally, the set of lower real numbers in $[\mathcal{C}^{op},\mathbf{Set}]$ is the presheaf
\begin{equation*}
\underline{\mathbb{R}}_{l}:\mathcal{C}^{op}\to\mathbf{Set},\ \ \underline{\mathbb{R}}_{l}(C)=OR((\downarrow C),\mathbb{R}).
\end{equation*}
Externally, the set of upper real numbers in $[\mathcal{C}^{op},\mathbf{Set}]$ is the presheaf
\begin{equation*}
\underline{\mathbb{R}}_{u}:\mathcal{C}^{op}\to\mathbf{Set},\ \ \underline{\mathbb{R}}_{u}(C)=OP((\downarrow C),\mathbb{R}).
\end{equation*}
The set of Dedekind real numbers $\underline{\mathbb{R}}_{d}$ of $[\mathcal{C}^{op},\mathbf{Set}]$ is externally given by the constant functor $\Delta(\mathbb{R})$.
\end{lem}

The next corollary described the value object of the contravariant model internally. It uses the following notation
\begin{equation*}
\forall\underline{x}\in\underline{\mathbb{R}}_{l}\ \forall\epsilon\in\underline{\mathbb{Q}}\ \ \underline{x}+\epsilon:=\{q+r\mid q<\underline{x},\ r<\epsilon\};
\end{equation*}
\begin{equation*}
\forall\overline{x}\in\underline{\mathbb{R}}_{u}\ \forall\epsilon\in\underline{\mathbb{Q}}\ \ \overline{x}+\epsilon:=\{q+r\mid q>\underline{x},\ r>\epsilon\}.
\end{equation*}
If we view a rational number $\epsilon$ as an upper or lower real number 
\begin{equation*}
\underline{\epsilon}=\{q\in\mathbb{Q}\mid q<\epsilon\},\ \ \overline{\epsilon}=\{q\in\mathbb{Q}\mid q>\epsilon\},
\end{equation*}
then $\underline{x}+\epsilon$ coincides with the sum $\underline{x}+\underline{\epsilon}$, and $\overline{x}-\epsilon$ coincides with both $\overline{x}+(\overline{-\epsilon})$, and $\overline{x}-\underline{\epsilon}$, where addition and subtraction are defined as in~\cite{vic2}

\begin{cor}
The presheaf (\ref{equ: valobj}) is the external description of the internal set
\begin{equation*}
\{(\overline{x},\underline{x})\in\underline{\mathbb{R}}_{u}\times\underline{\mathbb{R}}_{l}\mid\forall\epsilon\in\underline{\mathbb{Q}}^{+}\ \overline{x}-\epsilon<\underline{x}+\epsilon\},
\end{equation*}
where, $\overline{x}-\epsilon<\underline{x}+\epsilon$ means that  $(\overline{x}-\epsilon)\cap(\underline{x}+\epsilon)$ contains a rational number. 
\end{cor}

In what follows, we would like to view the upper and lower real numbers as internal spaces. Consider $\underline{\mathbb{R}}_{l}$, the internal set of real numbers in $[\mathcal{C}^{op}_{A},\mathbf{Set}]$. The corresponding \'etale bundle is given by $\pi_{l}:\mathcal{R}^{\downarrow}_{l,A}\to\mathcal{C}^{\downarrow}_{A}$, where
\begin{equation*}
\mathcal{R}^{\downarrow}_{l,A}=\{(C,s)\mid C\in\mathcal{C}, s\in OR((\downarrow C),\mathbb{R})\},
\end{equation*}
and $U\subseteq\mathcal{R}^{\downarrow}_{\,A}$ is open with respect to the \'etale topology iff
\begin{equation*}
\text{If}\ \ (C,s)\in U\ \ \text{and}\ \  D\subseteq C\ \ \text{then}\ \  (D,s|_{\downarrow D})\in U.
\end{equation*}

Provide $\mathcal{R}^{\downarrow}_{l,A}$ with the coarser topology generated by the \'etale opens
\begin{equation} \label{equ: basicR}
U_{x,C}=\{(D,s)\in\mathcal{R}^{\downarrow}_{l,A}\mid D\in(\downarrow C),\ s(D)>x\},\ \ C\in\mathcal{C}_{A}, x\in\mathbb{R}.
\end{equation}
Note that with respect to this topology, the function
\begin{equation*}
j:\mathcal{R}^{\downarrow}_{l,A}\to\mathcal{C}^{\downarrow}_{A}\times\mathbb{R}_{l},\ \ (C,s)\mapsto(C,s(C))
\end{equation*}
is a continuous map over $\mathcal{C}^{\downarrow}_{l,A}$, and the inverse image map $j^{-1}$ is an isomorphism of frames on the topologies. Whenever we want to see the internal lower reals as a topological space rather than a locale, we can use the bundle  $\pi_{l}:\mathcal{R}^{\downarrow}_{l,A}\to\mathcal{C}^{\downarrow}_{A}$, where $\mathcal{R}^{\downarrow}_{l,A}$ has the topology generated by (\ref{equ: basicR}). The same can be done for the upper reals using a topological space $\mathcal{R}^{\downarrow}_{u,A}$.

\subsection{Physical quantities as continuous maps} \label{subsec: physquant}

In this subsection we take an internal perspective on daseinised selfadjoint operators, by thinking of them as continuous functions from the space of states to the space of values. We have two reasons for this. The obvious one is that we are investigating the interplay between the internal language of the topoi and neorealism, i.e. formal proximity to classical structures. The second reason is that we want to investigate to what extent the elementary propositions $[a\in\Delta]$ can be obtained in an internal way. These propositions are labelled by opens $\Delta\in\mathcal{O}\mathbb{R}$ and have no obvious relation to the internal value object $\mathcal{R}$. Ideally, we would like to relate opens $\Delta\in\mathcal{O}\mathbb{R}$ to subobjects $\underline{\Delta}\subseteq\mathcal{R}$, such that for an operator $a$, represented by an arrow $\delta(a):\Sigma\to\mathcal{R}$, the elementary open $[a\in\Delta]$ is obtained internally as $\delta(a)^{-1}(\underline{\Delta})$. 

From Subsection~\ref{subsec: review} we know that for any $a\in A_{sa}$ and $C\in\mathcal{C}$, outer daseinisation provides an element $\delta^{o}(a)_{C}\in C_{sa}$. By Gelfand duality, we can see this as a continuous map 
\begin{equation*}
\widehat{\delta^{o}(a)_{C}}:\Sigma_{C}\to\mathbb{R}.
\end{equation*}
If $D\subseteq C$, then $\delta^{o}(a)_{D}\geq\delta^{o}(a)_{C}$ by definition of outer daseinisation. This entails
\begin{equation*}
\forall a\in A_{sa}\ \ \forall C\in\mathcal{C}\ \ \forall \lambda\in\Sigma_{C}\ \ \widehat{\delta^{o}(a)_{C}}(\lambda)\leq\widehat{\delta^{o}(a)_{D}}(\lambda|_{D}),
\end{equation*}
as this follows straight from
\begin{equation*}
\lambda(\delta^{o}(a)_{C})=\langle\delta^{o}(a)_{C},\lambda\rangle\leq\langle\delta^{o}(a)_{D},\lambda\rangle=\langle\delta^{o}(a)_{D},\lambda|_{D}\rangle=\lambda|_{D}(\delta^{o}(a)_{D}).
\end{equation*}
For a fixed $a\in A_{sa}$, and varying $C\in\mathcal{C}$, we can combine these maps into a single arrow $\underline{\delta^{o}(a)}:\underline{\Sigma}\to\underline{\mathbb{R}}^{\preceq}$. So internally $\underline{\delta^{o}(a)}$ defines a function from the spectral presheaf to the set of lower real numbers. The arrow is given by 
\begin{equation*}
\underline{\delta^{o}(a)}_{C}: \Sigma_{C}\to\text{OR}((\downarrow C),\mathbb{R}), \ \ \underline{\delta^{o}(a)}_{C}(\lambda)(D)=\langle\lambda,\delta^{o}(a)_{D}\rangle.
\end{equation*}

\begin{poe} \label{poe: outdas}
The function $\underline{\delta^{o}(a)}:\uS\to\underline{\mathbb{R}}_{l}$ is a continuous map of internal topological spaces.
\end{poe}

\begin{proof}
At the level of \'etale bundles, the natural transformation $\underline{\delta^{o}(a)}:\uS\to\underline{\mathbb{R}}_{l}$ is given by
\[ \xymatrix{
\Sigma^{\downarrow}_{A} \ar[rr]^{\delta^{o}(a)} \ar[rd]_{\pi} & & \mathcal{R}^{\downarrow}_{l,A} \ar[ld]^{\pi_{1}}\\
&\mathcal{C}^{\downarrow}_{A}} \]
where
\begin{equation*}
\delta^{o}(a)(C,\lambda)=(C,D\mapsto\langle\lambda,\delta^{o}(a)_{D}\rangle).
\end{equation*}
The function $\delta^{o}(a)$ is continuous with respect to the \'etale topologies, simply because it comes from a natural transformation, but we need to check that it is also continuous with respect to the coarser topologies, corresponding to the internal topologies. Consider the basic open $U_{x,C}$ of $\mathcal{R}^{\downarrow}_{l,A}$. Then
\begin{equation} \label{equ: wazzap}
\delta^{o}(a)^{-1}(U_{x,C})_{D}= \left\{
\begin{array}{ll}
\widehat{\delta^{o}(a)_{D}}^{-1}(x,+\infty) & \text{if } D\subseteq C\\
\emptyset & \text{if } D\nsubseteq C.
\end{array} \right.
\end{equation}
From (\ref{equ: wazzap}) it is clear that for each $D\in\mathcal{C}_{A}$, the set $\delta^{o}(a)^{-1}(U_{x,C})_{D}$ is open in $\Sigma_{D}$. Also, if $(D,\lambda)\in\delta^{o}(a)^{-1}(U_{x,C})$ and $D'\subseteq D$, then
\begin{equation*}
\langle\lambda|_{D'},\delta^{o}(a)_{D'}\rangle\geq\langle\lambda,\delta^{o}(a)_{D}\rangle>x,
\end{equation*}
so that $(D',\lambda|_{D'})\in\delta^{o}(a)^{-1}(U_{x,C})$. We conclude that $\delta^{o}(a)^{-1}(U_{x,C})$ is open in $\Sigma^{\downarrow}_{A}$ with respect to the topology generated by the closed open subobjects.
\end{proof}

Instead of continuous maps of spaces, we can view $\delta^{o}(a)$ as an internal map of locales by considering the commutative triangle
\[ \xymatrix{
\Sigma^{\downarrow}_{A} \ar[rr]^{\delta^{o}(a)} \ar[rd]_{\pi} & & \mathcal{C}^{\downarrow}_{A}\times\mathbb{R}_{l} \ar[ld]^{\pi_{1}}\\
&\mathcal{C}^{\downarrow}_{A}} \]

where $\delta^{o}(a):\Sigma\to\mathcal{C}\times\mathbb{R}_{l}$ is given by $(C,\lambda)\mapsto(C,\langle\lambda,\delta^{o}(a)_{C}\rangle)$. Under the identification of the category of locales in $[\mathcal{C}^{op},\mathbf{Set}]$ with the category of locales over $\mathcal{C}_{\downarrow}$, the triangle of locale maps over $\mathcal{C}_{\downarrow}$ corresponds to an internal locale map $\underline{\delta^{o}(a)}:\underline{\Sigma}\to\underline{\mathbb{R}}_{l}$.

Just like the presheaf of order-reversing functions, we can define the presheaf of order-preserving function $\underline{\mathbb{R}}^{\succeq}$. This presheaf coincides with presheaf of upper real numbers $\underline{\mathbb{R}}_{u}$. Inner daseinisation of a self-adjoint operator defines a natural transformation $\underline{\delta^{i}(a)}:\underline{\Sigma}\to\underline{\mathbb{R}}^{\succeq}$. We leave it to the reader to prove the following analogue of the previous proposition.

\begin{poe} \label{poe: indas}
The function  $\underline{\delta^{i}(a)}:\underline{\Sigma}\to\underline{\mathbb{R}}_{u}$ is a continuous map of internal topological spaces.
\end{poe}

\subsubsection{Covariant Version} \label{subsub: physquantco}

Before we try to connect the continuous daseinised operators to the elementary propositions we first look at the way this works in the covariant version of the topos approach. The topos $[\mathcal{C}_{A},\mathbf{Set}]$ is equivalent (even isomorphic) to the topos of sheaves over $\mathcal{C}^{\uparrow}_{A}$, the set $\mathcal{C}_{A}$ equipped with the upset Alexandroff topology. 

\begin{lem}
In $[\mathcal{C},\mathbf{Set}]$, the internal lower and upper reals (as sets) are externally given by the functors
\begin{equation*}
\underline{\mathbb{R}}_{l}:\mathcal{C}_{A}\to\mathbf{Set},\ \ \underline{\mathbb{R}}_{l}(C)=\text{OP}((\uparrow C),\mathbb{R}),
\end{equation*}
\begin{equation*}
\underline{\mathbb{R}}_{u}:\mathcal{C}_{A}\to\mathbf{Set},\ \ \underline{\mathbb{R}}_{u}(C)=\text{OR}((\uparrow C),\mathbb{R}).
\end{equation*}
\end{lem}

Note that with respect to the contravariant version the roles of order-preserving and order-reversing functions have been interchanged. In the covariant model, the role of the spectral presheaf, as a state space, is played by the internal Gelfand spectrum $\uS_{\uA}$ of $\uA$. Using the identification $[\mathcal{C},\mathbf{Set}]\cong Sh(\mathcal{C}_{\uparrow})$, and the observation that locales in $Sh(\mathcal{C}_{\uparrow})$ correspond to locale maps over $L(\mathcal{C}_{\uparrow})$, we describe the spectrum as a continuous map $\pi_{A}:\Sigma^{\uparrow}_{\uA}\to\mathcal{C}^{\uparrow}_{A}$ of topological spaces. Note that a map of spaces induces a map of the associated locales, so this would indeed define an internal locale. The external description of $\uS_{\uA}$ is given in the following proposition~\cite[Section 2.2]{wollie} .

\begin{poe} \label{prop: extspecco}
Let $\Sigma_{\uA}^{\uparrow}$ (or $\Sigma_{\uparrow}$ for short) be the set $\coprod_{C\in\mathcal{C}}\Sigma_{C}$ equipped with the topology, where $U\in\mathcal{O}\Sigma_{\uparrow}$ iff the following two conditions are satisfied
\begin{enumerate}
\item If $\lambda\in U_{D}$, $D\subseteq C$, and $\lambda'\in\Sigma_{C}$ such that $\lambda'|_{D}=\lambda$, then $\lambda'\in U_{C}$.
\item For every $C\in\mathcal{C}$, $U_{C}$ is open in $\Sigma_{C}$
\end{enumerate}
The continuous map
\begin{equation*}
\pi_{A}:\Sigma^{\uparrow}_{\uA}\to\mathcal{C}^{\uparrow}_{A}\ \ (C,\lambda)\mapsto C
\end{equation*}
is the external description of the spectrum $\uS_{\uA}$.
\end{poe}

Note that $\Sigma_{\uparrow}$ is the same set as $\Sigma_{\downarrow}$, but the topologies are different. In this covariant version the inner and outer daseinised operators define locale maps

\begin{poe}{(\cite[Proposition 6]{wollie}}
Outer daseinisation defines a commutative triangle of continuous maps
\[ \xymatrix{
\Sigma^{\uparrow}_{\uA} \ar[rr]^{\delta^{o}(a)} \ar[rd]_{\pi_{A}} & & \mathcal{C}^{\uparrow}_{A}\times\mathbb{R}_{u} \ar[ld]^{\pi_{1}}\\
&\mathcal{C}^{\uparrow}_{A},} \]
for which we denote the corresponding internal locale map as
\begin{equation*}
\underline{\delta^{o}(a)}:\underline{\Sigma}_{\underline{A}}\to\underline{\mathbb{R}}_{u}.
\end{equation*}
In the same way, inner daseinisation defines a locale map
\begin{equation*}
\underline{\delta^{i}(a)}: \underline{\Sigma}_{\underline{A}}\to\underline{\mathbb{R}}_{l}.
\end{equation*}
\end{poe}

At the level of sets and functions, this is the same triangle as for the contravariant version. The difference is only in the topologies. The same holds for inner daseinisation.

We can pair the two daseinisation maps together as the locale map
\begin{equation*}
\underline{\delta(a)}=\langle\underline{\delta^{i}(a)},\underline{\delta^{o}(a)}\rangle:\uS_{\uA}\to\underline{\mathbb{R}}_{l}\times\underline{\mathbb{R}}_{u},
\end{equation*}
which externally is described by

\[ \xymatrix{ \Sigma^{\uparrow}_{\uA} \ar[rr]^{\delta(a)} \ar[dr]_{\pi} & & \mathcal{C}^{\uparrow}_{A}\times\mathbb{R}_{l}\times\mathbb{R}_{u} \ar[dl]^{\pi_{1}} \\
& \mathcal{C}^{\uparrow}_{A} &} \]

where $\delta(a)(C,\lambda)=(C,\langle\lambda,\delta^{i}(a)_{C}\rangle,\langle\lambda,\delta^{o}(a)_{C}\rangle)$, and we used the identification
\begin{equation*}
(\mathcal{C}^{\uparrow}_{A}\times\mathbb{R}_{l})\times_{\mathcal{C}^{\uparrow}_{A}}(\mathcal{C}^{\uparrow}_{A}\times\mathbb{R}_{u})\cong\mathcal{C}^{\uparrow}_{A}\times\mathbb{R}_{l}\times\mathbb{R}_{u}.
\end{equation*}
The covariant approach normally uses the interval domain $\mathbb{IR}$ as a value object. In the topos $\mathbf{Set}$, as a set it has pairs $[x,y]$, with $x,y\in\mathbb{R}$, $x\leq y$, as elements. The topology on $\mathbb{IR}$ is generated by the basis
\begin{equation*}
(r,s)=\{[x,y]\in\mathbb{IR}\mid r<x\leq y<s\},\ \  r,s\in\mathbb{Q},\ \ r<s.
\end{equation*}
Consider the injective function
\begin{equation*}
j:\mathbb{IR}\to\mathbb{R}_{l}\times\mathbb{R}_{u}, \ \ j([x,y])=(x,y).
\end{equation*}
This function is continuous because $j^{-1}((r,+\infty]\times[-\infty,s))$ is equal to $(r,s)\in\mathcal{O}\mathbb{IR}$ if $r<s$, and the empty set if $r\geq s$. Note that for each context $C\in\mathcal{C}$ and any $\lambda\in\Sigma_{C}$, we have the inequality $\langle\lambda,\delta^{i}(a)_{C}\rangle\leq\langle\lambda,\delta^{o}(a)_{C}\rangle$ so the map $\delta(a)$ factors through the interval domain as

\[ \xymatrix{ \Sigma^{\uparrow}_{\uA} \ar[rr]^{\delta(a)} \ar[drr] & & \mathcal{C}^{\uparrow}_{A}\times\mathbb{R}_{l}\times\mathbb{R}_{u} \\
& & \mathcal{C}^{\uparrow}_{A}\times\mathbb{IR} \ar[u]_{\mathcal{C}\times j} } \]

Note that this is a commutative triangle in $\mathbf{Loc}/\mathcal{C}^{\uparrow}_{A}$, where $\mathcal{C}^{\uparrow}_{A}\times\mathbb{IR}$ is seen as a bundle over $\mathcal{C}^{\uparrow}_{A}$ by projecting on the first coordinate. This bundle $\pi_{1}:\mathcal{C}^{\uparrow}_{A}\times\mathbb{IR}\to\mathcal{C}^{\uparrow}_{A}$ is the external description of the interval domain $\underline{\mathbb{IR}}$ in $[\mathcal{C}_{A},\mathbf{Set}]$. The factorised map $\Sigma^{\uparrow}_{\uA}\to\mathcal{C}^{\uparrow}_{A}\times\mathbb{IR}$ is the external description of the daseinisation map $\underline{\Sigma}_{\uA}\to\underline{\mathbb{IR}}$ used in~\cite{wollie}.

Now we can connect this daseinisation map to the elementary propositions, at least for the case where we consider an open interval $\Delta=(r,s)$ in the set of (Dedekind) real numbers. We can translate this to an open subset of $\mathcal{C}^{\uparrow}_{A}\times\mathbb{IR}$ (or an open subset of $\mathcal{C}^{\uparrow}_{A}\times\mathbb{R}_{l}\times\mathbb{R}_{u}$) by
\begin{equation*}
\hat{\Delta}=\{(C,[x,y])\in\mathcal{C}^{\uparrow}_{A}\times\mathbb{IR}\mid C\in\mathcal{C}, [x,y]\in(r,s)\},
\end{equation*}
where we view $(r,s)$ as an open of $\mathbb{IR}$. In addition, define for any real number $\epsilon>0$
\begin{equation*}
\hat{\Delta}+\epsilon=\{(C,[x,y])\in\mathcal{C}^{\uparrow}_{A}\times\mathbb{IR}\mid C\in\mathcal{C}, [x,y]\in(r-\epsilon,s+\epsilon)\}.
\end{equation*}
For $\hat{\Delta}$, we get the corresponding open of $\Sigma^{\uparrow}_{\uA}$
\begin{equation} \label{equ: prop1}
\delta(a)^{-1}(\hat{\Delta})=\{(C,\lambda)\in\Sigma\mid r<\langle\lambda,\delta^{i}(a)_{C}\rangle\leq\langle\lambda,\delta^{o}(a)_{C}\rangle<s\}.
\end{equation}
In~\cite{wollie} the elementary proposition $[a\in\Delta]$, viewed externally as an open of $\Sigma_{\uparrow}$, was described as
\begin{equation} \label{equ: prop2}
[a\in\Delta]=\{(C,\lambda)\in\Sigma\mid\langle\lambda,\delta^{i}(\chi_{\Delta}(a))_{C}\rangle=1\},
\end{equation}
where $\chi_{\Delta}(a)$ is the spectral projection of $a$, associated with $\Delta$. The elementary proposition (\ref{equ: prop2}) was introduced to mimic the contravariant elementary proposition. 

The opens (\ref{equ: prop1}) and (\ref{equ: prop2}) are closely related.

\begin{tut}{(\cite[Lemma 3.11]{wollie})} \label{poe: das}
Let $r<s$ in $\mathbb{R}$, and $\epsilon>0$. Then
\begin{equation*}
\delta(a)^{-1}(\hat{\Delta})\subseteq[a\in\Delta]\subseteq\delta(a)^{-1}(\hat{\Delta}+\epsilon).
\end{equation*}
\end{tut}

This theorem establishes the relations between elementary propositions, described by inner daseinised projections, and two-sided daseinised self-adjoint operators. Furthermore, through this correspondence the external space of real numbers is linked to the internal value object in a straightforward way. The inclusions of the theorem follow from the identities
\begin{equation} \label{equ: ujelly1}
\langle\lambda,\delta^{i}(a)_{C}\rangle=\text{sup}\{r\in\mathbb{R}\mid\exists p\in\mathcal{P}(C),\ \langle\lambda,p\rangle=1,\ \ p\leq1-\chi_{[-\infty,r)}(a)\},
\end{equation}
\begin{equation} \label{equ: ujelly2}
\langle\lambda,\delta^{o}(a)_{C}\rangle=\text{inf}\{r\in\mathbb{R}\mid\exists p\in\mathcal{P}(C),\ \langle\lambda,p\rangle=1,\ \ p\leq\chi_{[-\infty,r)}(a)\},
\end{equation}
where $\mathcal{P}(C)$ is the Boolean algebra of projection operators of $C$. From these identities the connection to inner daseinisation of spectral projections of $a$ becomes clear. Note that for each projection operator $q\in\mathcal{P}(A)$ there exists an $p\in\mathcal{P}(C)$ with the properties $\langle\lambda,p\rangle=1$ and $p\leq q$, iff $\langle\lambda,\delta^{i}(q)_{C}\rangle=1$. So, for example, (\ref{equ: ujelly2}) can be written as
\begin{equation*}
\langle\lambda,\delta^{o}(a)_{C}\rangle=\text{inf}\{r\in\mathbb{R}\mid\langle\lambda,\delta^{i}(\chi_{[-\infty,r)}(a))_{C}\rangle=1\}.
\end{equation*}

\subsubsection{Contravariant Version}

As in the previous subsection, we can combine the two daseinisation maps into a single map $\mathcal{C}^{\downarrow}_{A}$;
\begin{equation*}
\delta(a):\Sigma_{A}^{\downarrow}\to\mathcal{C}^{\downarrow}_{A}\times\mathbb{R}_{u}\times\mathbb{R}_{l},\ \ (C,\lambda)\mapsto(C,\langle\lambda,\delta^{i}(a)_{C}\rangle,\langle\lambda,\delta^{o}(a)_{C}\rangle),
\end{equation*}
which is the external description of the internal continuous map
\begin{equation*}
 \underline{\delta(a)}=\langle\underline{\delta^{i}(a)},\underline{\delta^{o}(a)}\rangle:\underline{\Sigma}_{A}\to\underline{\mathbb{R}}_{u}\times\underline{\mathbb{R}}_{l}.
 \end{equation*}

Given $\Delta=(s,r)$, with $r,s\in\mathbb{R}$ such that $s<r$, we consider the open 
\begin{equation*}
\hat{\Delta}=\mathcal{C}_{A}\times[-\infty,r)\times(s,+\infty]\in\mathcal{O}(\mathcal{C}^{\downarrow}_{A}\times\mathbb{R}_{u}\times\mathbb{R}_{l}).
\end{equation*}
Likewise, if $\epsilon>0$, then $\hat{\Delta}+\epsilon$ is defined as $\mathcal{C}_{A}\times[-\infty,r+\epsilon)\times(s-\epsilon,+\infty]$.
\begin{align*}
\delta(a)^{-1}(\hat{\Delta}) &=\{(C,\lambda)\in\Sigma\mid r>\langle\lambda,\delta^{i}(a)_{C}\rangle\leq\langle\lambda,\delta^{o}(a)_{C}\rangle>s\},\\
&= \delta^{i}(a)^{-1}([-\infty,r))\cap\delta^{o}(a)^{-1}((s,+\infty]).
\end{align*}

\begin{tut} \label{tut: das}
Let $a\in A_{sa}$, $r,s\in\mathbb{R}$, $r<s$, and $\epsilon>0$. Then
\begin{equation*}
\delta(a)^{-1}(\hat{\Delta})\subseteq[a<r]\cap[a>s]\subseteq\delta(a)^{-1}(\hat{\Delta}+\epsilon).
\end{equation*}
\end{tut}

For the proof of Theorem~\ref{tut: das} we start by considering outer and inner daseinisation separately. First, recall  that the elementary propositions  of the contravariant model are given by
\begin{equation} \label{equ: propcon}
[a\in\Delta]=\{(C,\lambda)\in\Sigma\mid\langle\lambda,\delta^{o}(\chi_{\Delta}(a))_{C}\rangle=1\},
\end{equation}
where $a\in A_{sa}$, $\Delta\in\mathcal{O}\mathbb{R}$, and $\chi_{\Delta}(a)$ is the spectral projection associated to this pair. For half-intervals we will use the notation $[a<r]:=[a\in(-\infty,r)]$ and $[a>s]:=[a\in(s,+\infty)]$.

\begin{lem}
Let $r\in\mathbb{R}$, $a\in A_{sa}$, and $\delta^{i}(a):\Sigma^{\downarrow}_{A}\to\mathcal{C}^{\downarrow}_{A}\times\mathbb{R}_{u}$, be the corresponding (continuous) inner daseinised map. If we identify the half-interval $(-\infty,r)$ of the real numbers $\mathbb{R}$ with the open
\begin{equation*}
\mathcal{C}^{\downarrow}_{A}\times[-\infty,r)=\{(C,x)\in\mathcal{C}_{A}\times\mathbb{R}_{u}\mid x<r\}\in\mathcal{O}(\mathcal{C}^{\downarrow}_{A}\times\mathbb{R}_{u})
\end{equation*}
and write this open as $[-\infty,r)$ (with some abuse of notation), then, for each $\epsilon>0$,
\begin{equation}
\delta^{i}(a)^{-1}([-\infty,r))\subseteq [a<r]\subseteq \delta^{i}(a)^{-1}([\infty,r+\epsilon)).
\end{equation}
\end{lem}

\begin{proof}
Assume that $\lambda\in [a<r]_{C}$. By definition, this is equivalent to
\begin{equation} \label{equ: hihi}
\langle\lambda,\delta^{o}(\chi_{[-\infty,r)}(a))_{C}\rangle=1.
\end{equation}
By definition of outer daseinisation of projections this is in turn equivalent to
\begin{equation*}
\forall p\in\mathcal{P}(C),\ \ p\geq\chi_{[-\infty,r)}(a)\ \ \rightarrow\ \  \langle\lambda,p\rangle=1.
\end{equation*}
Switching to $\neg p=1-p$, this is equivalent to
\begin{equation*}
\forall p\in\mathcal{P}(C)\ \ p\leq 1-\chi_{[-\infty,r)}(a)\ \ \rightarrow\ \ \langle\lambda,p\rangle=0.
\end{equation*}
For any $x\geq r$
\begin{equation*}
1-\chi_{[-\infty,x)}(a)\leq1-\chi_{[-\infty,r)}(a).
\end{equation*}
Assume that for some projection $p\in\mathcal{P}(C)$, $p\leq1-\chi_{[-\infty,x)}(a)$. Then $p\leq1-\chi_{[-\infty,r)}(a)$, and by assumption $\langle\lambda, p\rangle=0$. We conclude that if $x\in\mathbb{R}$ is an element of the set
\begin{equation*}
\{y\in\mathbb{R}\mid\exists p\in\mathcal{P}(C)\ \ p\leq 1-\chi_{[-\infty,y)}(a)\ \wedge\ \langle\lambda,p\rangle=1\},
\end{equation*}
then $x<r$. As $\langle\lambda,\delta^{i}(a)_{C}\rangle$ is the supremum of such $x\in\mathbb{R}$ by (\ref{equ: ujelly1}), we know that $\langle\lambda,\delta^{i}(a)_{C}\rangle\leq r$. By definition, for each $\epsilon>0$, $\lambda\in\delta^{i}(a)^{-1}([-\infty,r+\epsilon))_{C}$. We have shown that
\begin{equation*}
[a<r]\subseteq\delta^{i}(a)^{-1}([-\infty,r+\epsilon)).
\end{equation*}
Next, assume that $\lambda\in\delta^{i}(a)^{-1}([-\infty,r))_{C}$. From (\ref{equ: ujelly1}) we deduce that 
\begin{equation*}
\forall p\in\mathcal{P}(C)\ \ p\geq\chi_{[-\infty,x)}(a)\ \wedge\ \langle\lambda, p\rangle=0\ \ \rightarrow x<r.
\end{equation*}
If $p\geq\chi_{[\-\infty,r)}(a)$, then $\langle\lambda,p\rangle=1$. This is equivalent to (\ref{equ: hihi}). We conclude that $\lambda\in[a<r]_{C}$, completing the proof of the lemma.
\end{proof}

\begin{lem}
Let $s\in\mathbb{R}$, $a\in A_{sa}$, and $\delta^{o}(a):\Sigma^{\downarrow}_{A}\to\mathcal{C}^{\downarrow}_{A}\times\mathbb{R}_{l}$, be the corresponding (continuous) outer daseinised map. If we identify the half-interval $(s,+\infty)$ of the real numbers $\mathbb{R}$ with the open
\begin{equation*}
\mathcal{C}^{\downarrow}_{A}\times(s,+\infty]=\{(C,x)\in\mathcal{C}_{A}\times\mathbb{R}_{l}\mid x>s\}\in\mathcal{O}(\mathcal{C}^{\downarrow}_{A}\times\mathbb{R}_{l}),
\end{equation*}
and write this open as $(s,+\infty]$ (with some abuse of notation). Then, for each $\epsilon>0$
\begin{equation} \label{equ: halfouter}
\delta^{o}(a)^{-1}((s,+\infty])\subseteq [a>s]\subseteq \delta^{o}(a)^{-1}((s-\epsilon,+\infty]).
\end{equation}
\end{lem}

\begin{proof}
Assume that $\lambda\in\delta^{o}(a)^{-1}((s,+\infty])_{C}$, implying $\langle\lambda,\delta^{o}(a)_{C}\rangle>s$. Define
\begin{equation*}
\epsilon_{0}=\frac{1}{2}(\langle\lambda,\delta^{o}(a)_{C}\rangle-s).
\end{equation*}
Let $p\in\mathcal{P}(C)$ satisfy $p\geq\chi_{(s,+\infty]}(a)$. Using
\begin{equation*}
p\geq\chi_{(s,+\infty]}(a)\geq1-\chi_{[-\infty,s+\epsilon_{0})}(a),
\end{equation*}
and (\ref{equ: ujelly2}), which tells us that
\begin{equation*}
\text{inf}\{x\in\mathbb{R}\mid\exists p\in\mathcal{P}(C)\ \ p\geq 1-\chi_{[-\infty,x)}(a)\ \wedge\ \langle\lambda,p\rangle=0\}>s+\epsilon_{0}.
\end{equation*}
We conclude that $\langle\lambda, p\rangle=1$. This implies that $\lambda\in[a>s]_{C}$, proving the left inequality of (\ref{equ: halfouter}). For the right inequality, assume that $\lambda\in[a>s]_{C}$. Let $x<s$, and assume that $p\geq1-\chi_{[-\infty,x)}(a)$. As $x<s$, we know 
\begin{equation*}
p\geq1-\chi_{[-\infty,x)}(a)\geq\chi_{(s,+\infty]}(a).
\end{equation*}
By assumption, this implies that $\langle\lambda,p\rangle=1$. This, in turn, implies
\begin{equation*}
\text{inf}\{x\in\mathbb{R}\mid\exists p\in\mathcal{P}(C)\ \ p\geq 1-\chi_{[-\infty,x)}(a)\ \wedge\ \langle\lambda,p\rangle=0\}\geq s.
\end{equation*}
By (\ref{equ: ujelly2}) $\langle\lambda,\delta^{o}(a)\rangle\geq s$, completing the proof of the lemma.
\end{proof}

Theorem~\ref{tut: das} follows from the previous two lemmas. Note that the inclusions
\begin{equation*}
\delta(a)^{-1}(\hat{\Delta})\subseteq[a<r]\cap[a>s]\subseteq\delta(a)^{-1}(\hat{\Delta}+\epsilon).
\end{equation*}
take the same shape as for the covariant version, given by Theorem~\ref{poe: das}, especially when we note that in the covariant version
\begin{equation*}
[a\in(s,r)]=[a<r]\cap[a>s].
\end{equation*}
This equality is consequence of the fact that inner daseinisation of projections preserves meets. Outer daseinisation does not preserves meets (it does preserve joins), so in the contravariant case we only have the inclusion
\begin{equation} \label{equ: diff}
[a\in(s,r)]\subseteq[a<r]\cap[a>s],
\end{equation}
where typically the inequality is strict. \\

Theorem~\ref{tut: das} tells us that at least in some cases, such as $\underline{[a<r]}$ and $\underline{[a>s]}$ where $a$ has a discrete spectrum , elementary propositions can be obtained internally as $\underline{\delta(a)}^{-1}(\underline{\Delta})$, where $\underline{\Delta}$ is suitably chosen open of the value space.

\subsection{States as Probability Valuations} \label{subsec: states}

In this subsection we investigate the connection between states, defined algebraically as normalised positive linear functionals on $A$, and, probability valuations on the spectral presheaf $\uS_{A}$, viewed internally to $[\mathcal{C}_{A}^{op},\mathbf{Set}]$ as an internal topological space. We can think of a probability valuation on a locale (or in particular a topological space) as a probability measure, but being defined only on the opens instead of the Borel algebra generated by it.

In the covariant approach states are described as internal probability valuations on the spectral locale $\uS_{\uA}$. Probability valuations on $\uS_{\uA}$ correspond bijectively with quasi-states on $A$ \cite{hls}. 
Closely related to this are the maps introduced by D\"oring in~\cite{doe3}, which he calls measures. In the setting of von Neumann algebras, it is straightforward to verify that such maps correspond bijectively with quasi-states. 

\begin{dork} \label{def: val}
Let $X$ be a locale in any topos $\mathcal{E}$, and let $[0,1]_{l}$ be the set of lower reals between 0 and 1. A \textbf{probability valuation} on $X$ is a function $\mu:\mathcal{O}X\to[0,1]_{l}$ satisfying the following conditions. Let $U,V\in\mathcal{O}X$ and $\{U_{\lambda}\}_{\lambda\in I}\subseteq\mathcal{O}X$ be a directed subset. Then
\begin{itemize}
\item $\mu$ is monotone. If $U\leq V$, then $\mu(U)\leq\mu(V)$.
\item $\mu(\perp)=0,\  \mu(\top)=1$, where $\perp$ and $\top$ are respectively the bottom and top element of $\mathcal{O}X$.
\item $\mu(U)+\mu(V)=\mu(U\wedge V)+\mu(U\vee V)$,
\item $\mu\left(\bigvee_{\lambda\in I}U_{\lambda}\right)=\bigvee_{\lambda\in I}\mu(U_{\lambda})$.
\end{itemize}
\end{dork}

\begin{lem} \label{lem: states}
A state $\psi:A\to\mathbb{C}$ defines a probability valuation $\mu_{\psi}$ on the spectral presheaf $\uS_{A}$.
\end{lem}

\begin{proof}
Using presheaf semantics we can describe what a probability valuation on $\uS_{A}$ comes down to externally for the topos $[\mathcal{C}_{A}^{op},\mathbf{Set}]$. Recall that 
\begin{equation*}
\underline{[0,1]}_{l}(C)\cong\text{OR}(\downarrow C,[0,1]). 
\end{equation*}
An internal probability valuation $\underline{\mu}:\mathcal{O}\uS_{A}\to\underline{[0,1]}_{l}$ is externally described by giving, for each $C\in\mathcal{C}$, a function
\begin{equation*}
\mu_{C}:\mathcal{O}\Sigma^{\downarrow}_{C}\to\text{OR}(\downarrow C,[0,1]),
\end{equation*}
such that, if $C\in\mathcal{C}$, and $U\in\mathcal{O}\Sigma^{\downarrow}_{C}$, then
\begin{equation*}
\forall D\in(\downarrow C)\ \  \mu_{C}(U)(D)=\mu_{D}(U\cap\Sigma^{\downarrow}_{D})(D).
\end{equation*}
We used the notation $\Sigma^{\downarrow}_{C}$ to denote $\coprod_{D\in(\downarrow C)}\Sigma_{D}$, equipped with the relative topology of $\Sigma^{\downarrow}_{A}$. The four axioms of Definition~\ref{def: val} translate externally to the following four conditions. For each $C\in\mathcal{C}$, and $D\in(\downarrow C)$,
\begin{itemize}
\item If $U\subseteq V$ in $\mathcal{O}\Sigma_{C}^{\downarrow}$, then $\mu_{C}(U)(D)\leq\mu_{C}(V)(D)$.
\item $\mu_{C}(\Sigma_{C}^{\downarrow})(D)=1$ and $\mu_{C}(\emptyset)(D)=0$.
\item If $U,V\in\mathcal{O}\Sigma^{\downarrow}_{C}$, then
\begin{equation*}
\mu_{C}(U)(D)+\mu_{C}(V)(D)=\mu_{C}(U\cap V)(D)+\mu_{C}(U\cup V)(D).
\end{equation*}
\item If $\{U_{\lambda}\}_{\lambda\in\Lambda}$ is a directed subset of $\mathcal{O}\Sigma_{C}^{\downarrow}$, then 
\begin{equation*}
\mu_{C}\left(\bigcup_{\lambda} U_{\lambda}\right)(D)=\sup_{\lambda}\left(\mu_{C}(U_{\lambda})(D)\right).
\end{equation*}
\end{itemize}
Let $\psi$ be a positive normalised linear functional on $A$, then $\psi$ defines such a valuation $\{\mu_{C}\}_{C\in\mathcal{C}}$ as follows. Restricting $\psi$ to $C\in\mathcal{C}$ gives a positive normalised linear functional $\psi|_{C}:C\to\mathbb{C}$. By the Riesz-Markov theorem this is equivalent to a probability valuation $\mu_{\psi}^{(C)}:\mathcal{O}\Sigma_{C}\to[0,1]$. Define
\begin{equation*}
(\mu_{\psi})_{C}:\mathcal{O}\Sigma^{\downarrow}_{C}\to\text{OR}(\downarrow C,[0,1]),\ \ (\mu_{\psi})_{C}(U)(D)=\mu_{\psi}^{(D)}(U_{D}),
\end{equation*}
where $U_{D}=U\cap\Sigma_{D}$. It is straightforward to verify that this definition satisfies all conditions required to define an internal probability valuation, and we leave this to the reader.
\end{proof}

\subsubsection{Does each probability valuation arise from a quasi-state?}

In the covariant topos model, the probability valuations on the spectral locale $\uS_{\uA}$ correspond bijectively with quasi-states on $A$. In particular, for von Neumann algebras without a type $I_{2}$ summand, probability valuations correspond to the states on $A$. We would like to know if each probability valuation on the spectral presheaf $\uS_{A}$ comes from a (quasi-)state on $A$. For the covariant model, the correspondence follows straight from the topos-valid version of the Riesz-Markov Theorem. But, in connection to the contravariant model, it may be more helpful to study probability valuations using presheaf semantics. 

\begin{dork}
A function $\psi:A\to\mathbb{C}$ is a \textbf{quasi-state} if it satisfies
\begin{itemize}
\item $\psi$ is positive; for each $a\in A$, $\psi(a^{\ast}a)\geq0$.
\item $\psi$ normalised; $\psi(1)=1$.
\item $\psi$ is quasi-linear; for each $C\in\mathcal{C}_{A}$, $\psi|_{C}$ is linear.
\item If $a,b\in A_{sa}$, then $\psi(a+ib)=\psi(a)+i\psi(b)$.
\end{itemize}
\end{dork}

\begin{poe}
In the covariant approach, probability valuations on $\uS_{\uA}$ correspond bijectively to quasi-states on $A$.
\end{poe}

\begin{proof}
If $\underline{\mu}$ is a probability valuation on the spectral locale, then for each $C\in\mathcal{C}$ it gives a function
\begin{equation} \label{equ: colocstate}
\mu_{C}:\mathcal{O}\Sigma^{\uparrow}_{C}\to\text{OP}(\uparrow C,[0,1]).
\end{equation}
If $C$ is a maximal context, then (\ref{equ: colocstate}) can be seen as a function
\begin{equation*}
\mu_{C}:\mathcal{O}\Sigma_{C}\to[0,1].
\end{equation*}
As $\underline{\mu}$ is a probability valuation, each such $\mu_{C}$ also satisfies the conditions for a probability valuation of $\Sigma_{C}$. This means that $\mu_{C}$ corresponds to a unique state $\psi_{C}:C\to\mathbb{C}$. These local states combine to a single quasi-state iff, given $D\subseteq C_{1}, C_{2}$, $\psi_{C_{1}}|_{D}=\psi_{C_{2}}|_{D}$. We proceed to show that this is indeed the case. If $D\subseteq C$, and $U\in\mathcal{O}\Sigma^{\uparrow}_{D}$, then
\begin{equation} \label{equ: natur}
\mu_{D}(U)(C)=\mu_{C}(U\cap\Sigma^{\uparrow}_{C})(C).
\end{equation}
Let $p\in D$ be a projection operator, corresponding to the closed open subset $S\subseteq\Sigma_{D}$. Define $(\uparrow S)$, an open of $\Sigma_{\uparrow}$, by taking $(\uparrow S)_{C}=\rho^{-1}_{CD}(S)$ if $C\supseteq D$ and $(\uparrow S)_{C}=\emptyset$ if $C\nsupseteq D$. The set-theoretic complement $S^{c}$ of $S$ in $\Sigma_{D}$ is open and closed in $\Sigma_{D}$ and also defines an open $\uparrow S^{c}$. Note that $(\uparrow S)\cap(\uparrow S^{c})=\emptyset$ and $(\uparrow S)\cup(\uparrow S^{c})=\Sigma^{\uparrow}_{D}$. This implies that for any $D\subseteq C$,
\begin{equation} \label{equ: cruxop}
\mu_{D}(\uparrow S)(D)+\mu_{D}(\uparrow S^{c})(D)=1=\mu_{D}(\uparrow S)(C)+\mu_{D}(\uparrow S^{c})(C).
\end{equation}
As $\mu_{D}(\uparrow S)$ and $\mu_{D}(\uparrow S^{c})$ are both order preserving with respect to $\uparrow D$, we conclude
\begin{equation} \label{equ: cruxopdeux}
\mu_{D}(\uparrow S)(C)=\mu_{D}(\uparrow S)(D).
\end{equation} 
Return to the situation $D\subseteq C_{1}, C_{2}$ with $C_{i}$ maximal, and recall that $p\in D$ is the projection operator corresponding to $S$ in $\Sigma_{D}$. Consequently, $p$ corresponds to $\rho^{-1}_{C_{i}D}(S)$ in $\Sigma_{C_{i}}$. We now compute
\begin{align*}
\psi_{C_{2}}(p) &=\mu_{C_{2}}(\rho^{-1}_{C_{2}D}(S))(C_{2})\\
&= \mu_{D}(\uparrow S)(C_{2})\\
&=\mu_{D}(\uparrow S)(D)\\
&= \mu_{D}(\uparrow S)(C_{1})\\
&=\mu_{C_{1}}(\rho^{-1}_{C_{1}D}(S))(C_{1})=\psi_{C_{1}}(p),
\end{align*}
where we used (\ref{equ: natur}) for the second and fifth equalities, and (\ref{equ: cruxopdeux}) for the third and fourth equalities. This proves that $\psi_{C_{2}}|_{D}=\psi_{C_{1}}|_{D}$, and demonstrates how internal valuations on the spectral locale can be identified with quasi-states. 
\end{proof}

Note that this proof cannot be directly applied to the contravariant model for probability valuations on the spectral presheaf. One obstacle is that restricting the valuation to a maximal context $C$ does not yield a probability valuation on $\Sigma_{C}$ (it is defined on $\mathcal{O}\Sigma^{\downarrow}_{C}$). Another obstacle is that for any given closed open $S$ of $\Sigma_{D}$, in general $(\downarrow S)\cap(\downarrow S^{c})\neq\emptyset$. \\

As far as the author knows, it is an open question whether there exist probability valuations on the spectral presheaf that do not arise from quasi-states.

\subsubsection{Pseudo-states}

Originally, in the contravariant model a state on $A$ was represented as a subobject $\mathfrak{w}\subseteq\uS_{A}$, called a pseudo-state. As alternatives, truth objects~\cite[Section 6.3]{di}, and later, the measures of~\cite{doe3} were introduced. We do not need to further discuss these alternatives as these lead to the same truth values as the subobjects $\mathfrak{w}$, when paired with elementary propositions.

In the assignment of a truth value to a state/proposition pair, the internal language of the topos $[\mathcal{C}_{A}^{op},\mathbf{Set}]$ plays a key role. Let $P$ denote some property of interest, e.g., one of the form $``a\in\Delta"$. This proposition is represented by a (clopen) subobject $\underline{P}\subseteq\uS_{A}$. Let $s$ be a state. In the contravariant approach, any state on $A$, in the sense of a positive normalised linear functional $\psi:A\to\mathbb{C}$, induces a subobject $\mathfrak{w}\subseteq\uS_{A}$. From the pair $(\psi,P)$, we can construct the following proposition $\mathfrak{w}\subseteq\underline{P}$ in the language of $[\mathcal{C}_{A}^{op},\mathbf{Set}]$. The associated truth value $[\mathfrak{w}\subseteq\underline{P}]:\underline{1}\to\underline{\Omega}$, is given by the sieve (i.e. downwards closed subset of $\mathcal{C}$)
\begin{equation} \label{equ: 2bveen2b}
[\mathfrak{w}\subseteq\underline{P}]=\{C\in\mathcal{C}\mid\psi(p_{C})=1\},
\end{equation}
where $p_{C}$ denotes the projection operators corresponding to the clopen subset $\underline{P}_{C}\subseteq\Sigma_{C}$. In particular, $\mathfrak{w}\subseteq\underline{P}$ is true at stage $C$ (i.e., $C\Vdash\mathfrak{w}\subseteq\underline{P}$) iff $\psi(p_{C})=1$.

\begin{rem}
If we want to stick close to classical physics, at least from the internal perspective, we can do the following. A state $s$ should correspond to an element of the state space $\uS_{A}$, so instead of thinking of $s$ as a subobject $\mathfrak{w}\subseteq\uS_{A}$, consider it as a generalised element $s:\mathfrak{w}\to\uS_{A}$. We think of the property $P$ as a subset of the state space $\uS_{A}$. This corresponds to an arrow $P: \underline{1}\to\mathcal{P}\uS_{A}$. In the language of $[\mathcal{C}^{op}_{A},\mathbf{Set}]$, we consider the term $s\in P$, represented by the arrow
\begin{equation*}
[s\in P]:\mathfrak{w}\cong\mathfrak{w}\times\underline{1}\stackrel{\langle s,P\rangle}{\longrightarrow}\uS\times\mathcal{P}\uS\stackrel{ev}{\longrightarrow}\underline{\Omega},
\end{equation*}
where $ev:\uS\times\mathcal{P}\uS\to\uO$ denotes the evaluation map of the exponential $\mathcal{P}\uS=\uO^{\uS}$. The relation between $[s\in P]:\mathfrak{w}\to\uO$ and $[\mathfrak{w}\subseteq\underline{P}]:\underline{1}\to\uO$ is as follows. The truth value $[\mathfrak{w}\subseteq\underline{P}]$ is equal to $\mathbf{true}:\underline{1}\to\uO$ (at stage $C$) iff $[s\in P]$ factors (at stage C) as
\[ \xymatrix{\mathfrak{w} \ar[rr]^{[s\in P]} \ar[dr]_{!} & & \uO \\
& \underline{1} \ar[ur]_{\mathbf{true}} } \]
\end{rem}

A state $\psi$ induces a probability valuation $\mu_{\psi}:\mathcal{O}\uS_{A}\to\underline{[0,1]}_{l}$ and a property $P$ is represented by an open of the space $\uS_{A}$, $\underline{P}:\underline{1}\to\mathcal{O}\uS_{A}$. For any $x\in[0,1]$, we can consider the proposition $\mu_{\psi}(\underline{P})\geq\underline{x}$. Here, $\underline{x}:\underline{1}\to\underline{[0,1]}_{l}$ is the constant function
\begin{equation*}
\underline{x}_{C}:(\downarrow C)\to[0,1],\ \ \underline{x}_{C}(D)=x.
\end{equation*}
Note that in this way we represent $x$ not only as a lower real number, but as a Dedekind real number as well, since the function $\underline{x}$ is constant. The truth value of this proposition is given by the sieve
\begin{equation} \label{equ: dadadabum}
[\mu_{\psi}(\underline{P})\geq\underline{x}]=\{C\in\mathcal{C}\mid \psi(p_{C})\geq x\}.
\end{equation}
In particular, for $x=1$ we get the same truth value as (\ref{equ: 2bveen2b}). 

\begin{rem}
If we represent states as internal probability valuations, we get the same truth values as normally used in the contravariant approach. One of the goals of the neorealism program was to get rid of probabilities altogether, by replacing them by generalised topos-theoretic truth values. In this paper we will not follow this interesting idea. More information can be found in~\cite{di5}. 
\end{rem}

\subsection{Summary}

By Proposition~\ref{prop: bundle} the spectral presheaf $\uS_{A}$ of the contravariant approach was considered as an internal topological space, where the topology was generated by the clopen subobjects. With respect to this topology, any self-adjoint element $a\in A_{sa}$ defined through daseinisation, a continuous map $\underline{\delta(a)}:\uS_{A}\to\underline{\mathbb{R}}_{u}\times\underline{\mathbb{R}}_{l}$, as shown by Propositions~\ref{poe: outdas} and~\ref{poe: indas}. Theorem~\ref{tut: das} showed that at least some of the elementary propositions $\underline{[a\in\Delta]}$ are of the form $\underline{\delta(a)}^{-1}(\underline{\hat{\Delta}})$, where $\underline{\hat{\Delta}}$ is a suitable open of $\underline{\mathbb{R}}_{u}\times\underline{\mathbb{R}}_{l}$. By Lemma~\ref{lem: states}, states $\psi$ on $A$ define probability valuations $\underline{\mu}_{\psi}$ on the space $\uS_{A}$.  The truth value of the proposition
\begin{equation*}
\underline{\mu}_{\psi}(\underline{[a\in\Delta]})=\underline{1},
\end{equation*}
is the same truth value normally assigned to the proposition $\underline{[a\in\Delta]}$ relative to the state $\psi$ in the contravariant approach.

Note that although we used the same objects, arrows and truth values as normally used for the contravariant approach, the perspective was different as we aimed for an internal view which was close to classical physics. 

There are some open questions. Are there probability valuations on the spectral presheaf that do not arise from quasi-states? And how deep to the connections between daseinised self-adjoint operators and the daseinisation of their spectral projections run?

\section{Differences from the Topos Set} \label{sec: differences}

Although the topos models $[\mathcal{C}^{op},\mathbf{Set}]$, and $[\mathcal{C},\mathbf{Set}]$ may resemble $\mathbf{Set}$ closely from the internal perspective, there are important differences, such as the logic  being multi-valued. In addition, the internal mathematics of these topos models validates neither the axiom of choice, nor the law of excluded middle. In this section we concentrate on these differences for both the contravariant and covariant approach.

\subsection{Contravariant Quantum Logic} \label{subsec: alternative}

In the two topos models, properties of the system under investigation, such as $[a\in\Delta]$ are represented by opens of $\mathcal{O}\Sigma_{\downarrow}$, or $\mathcal{O}\Sigma_{\uparrow}$. These two frames can be viewed as complete Heyting algebras. With this Heyting algebra structure,  $\mathcal{O}\Sigma_{\downarrow}$, or $\mathcal{O}\Sigma_{\uparrow}$ produce alternatives to the quantum logic of Birkhoff and von Neumann. At first glance these alternatives offered by the topos approaches look promising. In orthodox quantum logic, the lattice is non-distributive, making it hard to read $\wedge$ as \textit{and}, and $\vee$ as \textit{or}. Heyting algebras are always distributive. Another point is that orthodox quantum logic lacks a satisfactory implication operator, whereas a Heyting algebra has an implication by definition. In this respect the logics produced by $\mathcal{O}\Sigma_{\downarrow}$ and $\mathcal{O}\Sigma_{\uparrow}$ are looking good. However, we should realise that it is not a priori clear that the operations of these Heyting algebras $(\wedge,\vee,\neg,\to)$, have any physical significance. Consider the following simple example, which shows that recovering distributivity is not an achievement by itself.

Let $\mathcal{H}$ be a Hilbert space, and $\mathcal{P}\mathcal{H}$ be the power set of this space. Just as any power object in any topos defines a complete Heyting algebra, $\mathcal{P}\mathcal{H}$ defines a complete Heyting algebra, when ordered by inclusion. As we are working in $\mathbf{Set}$ it is even a complete Boolean algebra. Consider a proposition $[a\in\Delta]$. We associate to this proposition a projection operator $\chi_{\Delta}(a)$. Such a projection operator can be identified with a subset of $\mathcal{H}$ (which happens to be a closed subspace). In this way we represent elementary propositions $[a\in\Delta]$ as elements of a complete Boolean algebra, but the algebra $\mathcal{P}\mathcal{H}$ can hardly be called an interesting quantum logic. 

We don't expect the topos models to perform as badly as the logic $\mathcal{P}\mathcal{H}$, which completely ignores the linear structure of quantum theory. Even so, we need to investigate the Heyting algebras of the topos models.

Below, we try to understand the Heyting algebra structures of $\mathcal{O}\Sigma_{\downarrow}$, and $\mathcal{O}\Sigma_{\uparrow}$ by looking at the truth values these operations produce, when combined with states. 

In this subsection and the next one, we will \textit{assume that} $A$ \textit{is of the form} $M_{n}(\mathbb{C})$. This will make it easier to deal with the negation operation explicitly. It also implies that $\mathcal{O}\Sigma_{\downarrow}$ coincides with $\mathcal{O}_{cl}(\uS)$, the complete Heyting algebra typically considered in the contravariant model. 

We start with the Heyting algebra $\mathcal{O}\Sigma_{\downarrow}$ of the contravariant model, and treat the covariant version $\mathcal{O}\Sigma_{\uparrow}$ in the next subsection.

\subsubsection{Single Proposition}

Consider an elementary proposition $[a\in\Delta]$. We represent such a proposition as an open of $\mathcal{O}\Sigma_{\downarrow}$ by taking the outer daseinisation of the spectral projection $\chi_{\Delta}(a)$. If $\Delta$ is an open half-interval, then by Theorem~\ref{tut: das}, this $[a\in\Delta]$ is equal to $\delta(a)^{-1}(\hat{\Delta})$, for a suitably chosen $\hat{\Delta}$.

Let $\psi$ be a state on $A$, and let $\underline{\mu}:\mathcal{O}\uS\to\underline{[0,1]}_{l}$ be the internal probability valuations associated to it as in Subsection~\ref{subsec: states}. The elementary proposition $[a\in\Delta]$ defines an open $\underline{[a\in\Delta]}:\underline{1}\to\mathcal{O}\uS$. Consider the internal proposition (in the sense of a closed formula)
\begin{equation*}
\underline{\mu}(\underline{[a\in\Delta]})=\underline{1}.
\end{equation*}
By (\ref{equ: dadadabum}), this proposition is true at stage $C$ iff
\begin{equation*}
\psi(\delta^{o}(\chi_{\Delta}(a))_{C})=1,
\end{equation*}
which, by spelling out the definition of outer daseinisation of projections, is
\begin{equation*}
\psi\left(\bigwedge\{p\in\text{Proj}(C)\mid p\geq\chi_{\Delta}(a)\right)=1.
\end{equation*}
This leads to the following proposition.

\begin{poe}
For any $a\in A_{sa}$ and $\Delta\in\mathcal{O}\mathbb{R}$, and state $\psi$ the following two conditions are equivalent:
\begin{enumerate}
\item $C\Vdash \underline{\mu}((\underline{[a\in\Delta]})=\underline{1})$;
\item $\forall p\in\text{Proj}(C)\ \ p\geq\chi_{\Delta}(a)\ \rightarrow \psi(p)=1$.
\end{enumerate}
\end{poe}

The proposition $[a\in\Delta]$ is true, relative to a state $\psi$, and in context $C$, iff by performing only the measurements allowed by $C$, it is impossible to refute; given the system, prepared in the state $\psi$, a measurement of $a$ yields a value in $\Delta$ with certainty.  Of course, if we want to completely avoid operationalist notions, then this does not yield a satisfactory account of truth. 

\subsubsection{Disjunction}

Let $a_{1},a_{2}\in A_{sa}$, and $\Delta_{1},\Delta_{2}\in\mathcal{O}\mathbb{R}$. In order to obtain an understanding of the clopen subobject
\begin{equation*}
\underline{[a_{1}\in\Delta_{1}]}\vee\underline{[a_{2}\in\Delta_{2}]},
\end{equation*}
we pick an arbitrary state $\psi$ and consider the truth value of the proposition
\begin{equation*}
\underline{\mu}(\underline{[a_{1}\in\Delta_{1}]}\vee\underline{[a_{2}\in\Delta_{2}]})=\underline{1},
\end{equation*}
where $\underline{\mu}=\underline{\mu}_{\psi}$ is internal probability valuation corresponding to $\psi$. This proposition is true at stage $C\in\mathcal{C}$ (equivalently, the sieve of the truth value of the proposition contains $C$) iff
\begin{equation*}
\mu_{C}(\underline{[a_{1}\in\Delta_{1}]}_{C}\cup\underline{[a_{2}\in\Delta_{2}]}_{C})=1.
\end{equation*}
By definition of the local valuation $\mu_{C}$, this simply states that,
\begin{equation*}
\psi(\delta^{o}(\chi_{\Delta_{1}}(a_{1}))_{C}\vee\delta^{o}(\chi_{\Delta_{2}}(a_{2}))_{C})=1.
\end{equation*}
Recall that outer daseinisation of projections respects $\vee$, giving the simplification
\begin{equation*}
\psi(\delta^{o}(\chi_{\Delta_{1}}(a_{1})\vee\chi_{\Delta_{2}}(a_{2}))_{C})=1.
\end{equation*}
Spelling out the definition of outer daseinisation of projections, this is equivalent to
\begin{equation*}
\forall p\in\text{Proj}(C)\ \ p\geq\chi_{\Delta_{1}}(a_{1})\vee\chi_{\Delta_{2}}(a_{2})\ \rightarrow \psi(p)=1.
\end{equation*}
We collect this result in the following proposition.
\begin{poe} \label{lem: discontra}
If we define $\underline{[a\in\Delta]}$ using $\delta^{o}(\chi_{\Delta}(a))$, then the following two conditions are equivalent in the contravariant model:
\begin{enumerate}
\item $C\Vdash\ \underline{\mu}(\underline{[a_{1}\in\Delta_{1}]}\vee\underline{[a_{2}\in\Delta_{2}]})=\underline{1}$;
\item If $p\in\text{Proj}(C)$ satisfies $p\geq\chi_{\Delta_{1}}(a_{1})$ and $p\geq\chi_{\Delta_{2}}(a_{2})$, then $\psi(p)=1$.
\end{enumerate}
\end{poe}
We could interpret the result of this proposition in the following way: the internal proposition $\underline{[a_{1}\in\Delta_{1}]}\vee\underline{[a_{2}\in\Delta_{2}]}$ is true at context $C$ iff by using a single measurement allowed by $C$ it is impossible to refute both claims: for the system in state $\psi$, a measurement of $a_{i}$ yields a value in $\Delta_{i}$ with certainty, where $i\in\{1,2\}$.

\subsubsection{Conjunction}

At least on the mathematical level, the truth values from Proposition~\ref{lem: discontra} take on a simple form. This is a consequence of the fact that outer daseinisation respects joins of projection operators. How do the conjunctions fare? Consider the truth value of the proposition
\begin{equation*}
\underline{\mu}(\underline{[a_{1}\in\Delta_{1}]}\wedge\underline{[a_{2}\in\Delta_{2}]})=\underline{1}.
\end{equation*}
This proposition is true at stage $C$ iff
\begin{equation}  \label{equ: gen1}
\psi(\delta^{o}(\chi_{\Delta_{1}}(a_{1}))_{C}\wedge\delta^{o}(\chi_{\Delta_{2}}(a_{2}))_{C})=1.
\end{equation}
Spelling out the definition of outer daseinisation, and using distributivity of meets, this is equivalent to
\begin{equation*}
\psi\left(\bigwedge\{p\in\text{Proj}(C)\mid p\geq\chi_{\Delta_{1}}(a_{1})\ \text{or}\ p\geq\chi_{\Delta_{2}}(a_{2})\}\right)=1
\end{equation*}
Note that this identity implies
\begin{equation} \label{equ: gen2}
\forall p\in\text{Proj}(C)\ p\geq\chi_{\Delta_{1}}(a_{1})\ \text{or}\ p\geq\chi_{\Delta_{2}}(a_{2})\ \rightarrow \ \psi(p)=1.
\end{equation}
Next, assume (\ref{equ: gen2}). It follows that $\psi(\delta^{o}(\chi_{\Delta_{i}}(a_{i}))_{C})=1$, where $i\in\{1,2\}$. Let the clopens $S_{i}$ of $\Sigma_{C}$ correspond to the projections $\delta^{o}(\chi_{\Delta_{i}}(a_{i}))_{C}$, and let $\mu:\mathcal{O}\Sigma_{C}\to[0,1]$ denote the probability valuation corresponding to the state $\psi|_{C}$. By assumption $\mu(S_{i})=1$, for $i\in\{1,2\}$. The modular law implies
\begin{equation*}
\mu(S_{1}\cap S_{2})=\mu(S_{1})+\mu(S_{2})-\mu(S_{1}\cup S_{2})=1,
\end{equation*}
which in turn implies (\ref{equ: gen1}).

\begin{poe} \label{lem: concontra}
If we define $\underline{[a\in\Delta]}$ using $\delta^{o}(\chi_{\Delta}(a))$, then the following two conditions are equivalent in the contravariant model:
\begin{enumerate}
\item $C\Vdash\ \underline{\mu}(\underline{[a_{1}\in\Delta_{1}]}\wedge\underline{[a_{2}\in\Delta_{2}]})=\underline{1}$;
\item If $p\in\text{Proj}(C)$ satisfies $p\geq\chi_{\Delta_{1}}(a_{1})$ or $p\geq\chi_{\Delta_{2}}(a_{2})$, then $\psi(p)=1$.
\end{enumerate}
\end{poe}

We could interpret the result of this proposition in the following way: the internal proposition $\underline{[a_{1}\in\Delta_{1}]}\wedge\underline{[a_{2}\in\Delta_{2}]}$ is true at context $C$ iff using only measurements of $C$, we cannot refute either claim: A measurement of $a_{1}$ yields a value in $\Delta_{1}$ with certainty, and, a measurement of $a_{2}$ yields a value in $\Delta_{2}$ with certainty.

\subsubsection{Negation}

Negation in $\mathcal{O}\Sigma^{\downarrow}$ is more complicated than conjunction and disjunction. In order to describe it, we use the following notation. If $p\in C$ is a projection operator, then $S^{C}_{p}$ denotes the corresponding clopen subset of $\Sigma_{C}$. The superscript $C$ is added to distinguish between $S^{D}_{p}$ and $S^{C}_{p}$, whenever $p\in D\subseteq C$. 

The negation of $\mathcal{O}\Sigma^{\downarrow}$ is given by
\begin{equation*}
(\neg\underline{[a\in\Delta]})_{C}=\{\lambda\in\Sigma_{C}\mid\forall D\subseteq C\ \ \lambda|_{D}\notin\underline{[a\in\Delta]}_{D}\}.
\end{equation*}
This is more conveniently written as
\begin{equation*}
(\neg\underline{[a\in\Delta]})_{C}=\bigcap_{D\subseteq C}\rho^{-1}_{CD}(S^{D}_{\delta^{o}(\chi_{\Delta}(a))_{D}})^{co},
\end{equation*}
where the superscript $co$ denotes the set-theoretic complement. For any $p\in D\subseteq C$, we have $\rho^{-1}_{CD}(S^{D}_{p})=S^{C}_{p}$. By this observation, the previous expression simplifies to
\begin{equation*}
(\neg\underline{[a\in\Delta]})_{C}=\bigcap_{D\subseteq C}(S^{C}_{\delta^{o}(\chi_{\Delta}(a))_{D}})^{co}.
\end{equation*}
We can get rid of the set-theoretic complement by using the relation
\begin{equation*}
\forall C\in\mathcal{C}\ \forall p\in\text{Proj}(A)\ \ \delta^{i}(1-p)_{C}=1-\delta^{o}(p)_{C};
\end{equation*}
see e.g.~\cite[(5.59)]{di}. At the level of the Gelfand spectra, this translates to
\begin{equation*}
(S^{C}_{\delta^{o}(p)_{D}})^{co}=S^{C}_{\delta^{i}(1-p)_{D}}.
\end{equation*}
We deduce
\begin{equation*}
(\neg\underline{[a\in\Delta]})_{C}=\bigcap_{D\subseteq C}S^{C}_{\delta^{i}(1-\chi_{\Delta}(a))_{D}}=S^{C}_{\bigwedge_{D}\delta^{i}(1-\chi_{\Delta}(a))_{D}}.
\end{equation*}
Using this identity, we find that
\begin{equation*}
C\Vdash \underline{\mu}(\neg\underline{[a\in\Delta]})=\underline{1}
\end{equation*}
is equivalent to
\begin{equation} \label{equ: wtf}
\psi\left(\bigwedge_{D\in(\downarrow C)}\delta^{i}(1-\chi_{\Delta}(a))_{D}\right)=1.
\end{equation}
Assume that the intersection of $\Delta$ with the set of eigenvalues of $a$ is non empty (i.e., $\chi_{\Delta}(a)\neq0$), then condition (\ref{equ: wtf}) cannot be satisfied. This simply follows from the observation that the inner daseinisation of $1-\chi_{\Delta}(a)$ with respect to the trivial context $\mathbb{C}$ is equal to $0$. However, if we remove the bottom element $\mathbb{C}$ from $\mathcal{C}$, things become more interesting, at least mathematically. As we shall see, in this setting, the context $C$ needs to satisfy strong conditions in order for (\ref{equ: wtf}) to hold. In what follows we use the notation $q:=1-\chi_{\Delta}(a)$.

Let $p\in\text{Proj}(C)$ have the property that neither $p\leq q$, nor $1-p\leq q$. If $D\subseteq C$ is the context generated by $p$, then $\delta^{i}(q)_{D}=0$. As $\psi(0)=0$, we conclude that a necessary condition for (\ref{equ: wtf}) to hold is
\begin{equation*}
\forall p\in\text{Proj}(C)\ \ \text{either}\ \ p\leq q\ \ \text{or}\ \ 1-p\leq q.
\end{equation*}
Note that $q<1$ by assumption, so only one of the two options can hold. Also note that this condition implies that $C$ commutes with $q$. As we are working with matrix algebras $A=M_{n}(\mathbb{C})$ in this subsection, we can find projections $p_{1},\ldots,p_{k}$ in $C$ such that $p_{i}\cdot p_{j}=0$ if $i\neq j$ and $\sum_{i=1}^{k}p_{i}=1$. If (\ref{equ: wtf}) holds, we can sort these projections as follows. The set $L=\{p_{1},\ldots,p_{l}\}$ consists of the $p_{i}$ such that $p_{i}\leq q$. The set $R=\{p_{l+1},\ldots,p_{k}\}$ consists of the $p_{j}$ such that $1-p_{j}\leq q$. Note that $L\cap R=\emptyset$, and both sets are non empty. Also note that $\delta^{i}(q)_{C}=p_{1}+\ldots+p_{l}$.

Assume that $L$ has at least two elements. Let $D_{1}$, and $D_{2}$ be the context generated by the projections
\begin{equation*}
D_{1}=\{p_{1}+p_{l+1},p_{2},\ldots,p_{l},p_{l+2},\ldots,p_{k}\}'';
\end{equation*}
\begin{equation*}
D_{2}=\{p_{1},p_{2}+\ldots+p_{l}+p_{l+1}, p_{l+2},\ldots, p_{k}\}''.
\end{equation*}
Then $\delta^{i}(q)_{D_{1}}=p_{2}+\ldots+p_{l}$ and $\delta^{i}(q)_{D_{2}}=p_{1}$. We conclude that 
\begin{equation*}
\delta^{i}(q)_{D_{1}}\wedge\delta^{i}(q)_{D_{2}}=0,
\end{equation*}
and (\ref{equ: wtf}) cannot be satisfied. If  (\ref{equ: wtf}) holds, then $L$ is a singleton. In an analogous way it can be shown that $R$ contains exactly one element. This implies that the projection lattice of $C$ must be of the form $\{0, p, 1-p, 1\}$, with either $p\leq q$, or $1-p\leq q$.

\begin{poe}
If we define $\underline{[a\in\Delta]}$ using $\delta^{o}(\chi_{\Delta}(a))$, and remove $\mathbb{C}$ from $\mathcal{C}$, then the following two conditions are equivalent in the contravariant model:
\begin{enumerate}
\item $C\Vdash \underline{\mu}(\neg\underline{[a\in\Delta]})=\underline{1}$;
\item There exists a projection $p\in C$, that generates $C$, and satisfies $p\geq\chi_{\Delta}(a)$, as well as $\psi(p)=0$.
\end{enumerate}
\end{poe}

Only the coarsest contexts that commute with $\chi_{\Delta}(a)$ contribute to the truth value of $\underline{\mu}(\neg\underline{[a\in\Delta]})=\underline{1}$. This emphasis on coarser contexts makes it hard to find a physical interpretation of the negation operation, if this is possible at all. We will also encounter this problem with the more general Heyting implication. This problem may suggest that it is a mistake to seek an interpretation of the contravariant quantum logic in terms of refutation, as we did for conjunction and disjunction, but what alternatives are there? Unfortunately, it seems that negation and implication, as natural as they are from a topos theoretic perspective, do not seem to have a clear physical motivation. 

\subsection{Covariant Quantum Logic}

We continue with the complete Heyting algebra $\mathcal{O}\Sigma_{\uparrow}$ of the covariant model. As in the previous subsection we restrict to matrix algebras $A=M_{n}(\mathbb{C})$. The elementary proposition $[a\in\Delta]$ will be represented by taking the inner daseinisation of the spectral projection $\chi_{\Delta}(a)$. As long as $\Delta$ is an open interval or open half-interval, $[a\in\Delta]$ is also of the form $\delta(a)^{-1}(\hat{\Delta})$, for an appropriate $\hat{\Delta}$.

\subsubsection{Single Proposition}

The elementary proposition $[a\in\Delta]$ defines an open $\underline{[a\in\Delta]}:\underline{1}\to\mathcal{O}\uS_{\uA}$. Relative to a state $\psi$, represented by an internal probability valuation $\underline{\mu}$, we will study
\begin{equation*}
C\Vdash\underline{\mu}(\underline{[a\in\Delta]})=\underline{1}.
\end{equation*}
This is equivalent to
\begin{equation*}
\psi(\delta^{i}(\chi_{\Delta}(a))_{C})=1,
\end{equation*}
Recall that $\delta^{i}(\chi_{\Delta}(a))_{C}$ is the largest projection of $C$ that is smaller than $\chi_{\Delta}(a)$. 

\begin{poe}
For any $a\in A_{sa}$ and $\Delta\in\mathcal{O}\mathbb{R}$, and state $\psi$ the following two conditions are equivalent in the covariant model:
\begin{enumerate}
\item $C\Vdash \underline{\mu}((\underline{[a\in\Delta]})=\underline{1})$;
\item $\exists p\in\text{Proj}(C)\ \ p\leq\chi_{\Delta}(a)\ \ \text{and}\ \ \psi(p)=1$.
\end{enumerate}
\end{poe}

Truth of $[a\in\Delta]$ relative to the state $\psi$ and context $C$ holds iff $C$ provides us with a measurement with which we can affirm that the system, prepared in state $\psi$, upon a measurement of $a$, yields a value in $\Delta$ with certainty.

\subsubsection{Conjunction}

Our treatment of the covariant conjunction operation resembles that of the contravariant disjunction. Let $a_{1},a_{2}\in A_{sa}$, and $\Delta_{1},\Delta_{2}\in\mathcal{O}\mathbb{R}$. Consider
\begin{equation*}
C\Vdash\underline{\mu}(\underline{[a_{1}\in\Delta_{1}]}\wedge\underline{[a_{2}\in\Delta_{2}]})=\underline{1},
\end{equation*}
where $\underline{\mu}$ is internal probability valuation corresponding to $\psi$. This condition is equivalent to
\begin{equation*}
\psi(\delta^{i}(\chi_{\Delta_{1}}(a_{1}))_{C}\wedge\delta^{i}(\chi_{\Delta_{2}}(a_{2}))_{C})=1.
\end{equation*}
Recall that inner daseinisation of projections respects $\wedge$, giving the simplification
\begin{equation*}
\psi(\delta^{i}(\chi_{\Delta_{1}}(a_{1})\wedge\chi_{\Delta_{2}}(a_{2}))_{C})=1.
\end{equation*}
As in the single proposition case, this amounts to
\begin{equation*}
\exists p\in\text{Proj}(C)\ \ p\leq\chi_{\Delta_{1}}(a_{1})\wedge\chi_{\Delta_{2}}(a_{2})\ \rightarrow \psi(p)=1.
\end{equation*}
\begin{poe} \label{lem: conco}
If we define $\underline{[a\in\Delta]}$ using $\delta^{i}(\chi_{\Delta}(a))$, then the following two conditions are equivalent in the covariant model:
\begin{enumerate}
\item $C\Vdash\ \underline{\mu}(\underline{[a_{1}\in\Delta_{1}]}\wedge\underline{[a_{2}\in\Delta_{2}]})=\underline{1}$;
\item $\exists p\in\text{Proj}(C)$ such that $p\leq\chi_{\Delta_{1}}(a_{1})$, $p\leq\chi_{\Delta_{2}}(a_{2})$, and $\psi(p)=1$.
\end{enumerate}
\end{poe}
Truth of the meet of elementary propositions, relative to a state $\psi$, and context $C$ is therefore equivalent to: there is a measurement allowed by $C$, by which we can affirm that for the system, when prepared in the state $\psi$, a measurement of $a_{1}$ would yield a value in $\Delta_{1}$ with certainty, and for such a system, a measurement of $a_{2}$ would yield a value in $\Delta_{2}$ with certainty.

\subsubsection{Disjunction}

Our treatment of the covariant disjunction reminds us of the contravariant conjunction. Consider the forcing relation
\begin{equation*}
C\Vdash\underline{\mu}(\underline{[a_{1}\in\Delta_{1}]}\vee\underline{[a_{2}\in\Delta_{2}]})=\underline{1},
\end{equation*}
or, equivalently,
\begin{equation*}
\psi(\delta^{i}(\chi_{\Delta_{1}}(a_{1}))_{C}\vee\delta^{i}(\chi_{\Delta_{2}}(a_{2}))_{C})=1.
\end{equation*}
Spelling out the definition of inner daseinisation, and using distributivity of joins, this is equivalent to
\begin{equation} \label{equ: geno1}
\psi\left(\bigvee\{p\in\text{Proj}(C)\mid p\leq\chi_{\Delta_{1}}(a_{1})\ \text{or}\ p\leq\chi_{\Delta_{2}}(a_{2})\}\right)=1.
\end{equation}
Note that this identity is implied by,
\begin{equation} \label{equ: geno2}
\exists p\in\text{Proj}(C)\ p\leq\chi_{\Delta_{1}}(a_{1})\ \text{or}\ p\leq\chi_{\Delta_{2}}(a_{2})\ \ \text{and}\ \ \psi(p)=1.
\end{equation}
Note that (\ref{equ: geno1}) and (\ref{equ: geno2}) are \textbf{not} equivalent. This is because for a pair $p_{1},p_{2}\in\text{Proj}(C)$, it is possible that $\psi(p_{1}\vee p_{2})=1$, whilst neither $\psi(p_{1})=1$, nor $\psi(p_{2})=1$. The forcing relation is weaker than the affirmation of one of the two claims: given the system, prepared in state $\psi$, a measurement of $a_{i}$ yields a value in $\Delta_{i}$ with certainty ($i\in\{1,2\}$).

\subsubsection{Negation}

For matrix algebras, the negation of $\mathcal{O}\Sigma_{\uparrow}$ was first described in~\cite{chls} in terms of projections. There it was shown that the open subset $(\neg[a\in\Delta])_{C}$ of $\Sigma_{C}$ corresponds to the projection
\begin{equation*}
\bigvee\{p\in\text{Proj}(C)\mid\forall E\in(\uparrow C)\ \ p\leq1-\delta^{i}(p)_{E}\}.
\end{equation*}
Using 
\begin{equation*}
1-\delta^{i}(\chi_{\Delta}(a))_{E}=\delta^{o}(1-\chi_{\Delta}(a))_{E}=\delta^{o}(\chi_{\mathbb{R}-\Delta}(a))_{E},
\end{equation*}
we find
\begin{equation*}
(\neg[a\in\Delta])_{C}=\bigcap_{E\supseteq C}\{\lambda\in\Sigma_{C}\mid\rho^{-1}_{EC}(\lambda)\subseteq S^{E}_{\delta^{o}(\chi_{\mathbb{R}-\Delta}(a))_{E}}\}
\end{equation*}
This complicated expression makes it hard to understand the condition
\begin{equation*}
C\Vdash\underline{\mu}(\neg\underline{[a\in\Delta]})=\underline{1}.
\end{equation*}
However, if we restrict our attention to maximal contexts, then the negation simplifies considerably. 
The forcing relation is satisfied iff
\begin{equation*}
\psi(\delta^{o}(\chi_{\mathbb{R}-\Delta}(a))_{C})=1,
\end{equation*}
which is equivalent to
\begin{equation*}
\forall p\in\text{Proj}(C)\ \ p\geq\chi_{\mathbb{R}-\Delta}(a)\ \to\ \psi(p)=1.
\end{equation*}

\begin{poe}
If we define $\underline{[a\in\Delta]}$ using $\delta^{i}(\chi_{\Delta}(a))$, and consider a maximal context $C$ in $\mathcal{C}$, then the following two conditions are equivalent in the covariant model:
\begin{enumerate}
\item $C\Vdash \underline{\mu}(\neg\underline{[a\in\Delta]})=\underline{1}$;
\item For each projection $p\in C$, if $p\leq\chi_{\Delta}(a)$, then $\psi(p)=0$;
\item For each projection $p\in C$, if $p\geq\chi_{\mathbb{R}-\Delta}(a)$, then $\psi(p)=1$.
\end{enumerate}
\end{poe}

Hence, using only measurements allowed by the maximally refined context $C$ we cannot refute the claim that the system, when prepared in the state $\psi$, upon a measurement of $a$ yields a value outside of $\Delta$ with certainty.

If $C$ is not maximal, then $C\Vdash\underline{\mu}(\neg\underline{[a\in\Delta]})=\underline{1}$ implies that for any refinement $E$ of $C$ (i.e. $E\supseteq C$) we cannot refute the aforementioned claim using only measurements from the context $E$. As for the contravariant model, the physical content of the negation operation seems questionable.

\subsubsection{Discussion}

Guided by the truth values obtained from state-proposition pairs, it seems tempting to read the logic of the contravariant model as a logic of refutation and the logic of the covariant logic as one of affirmation. Through the correspondence $\delta^{o}(p)_{C}=1-\delta^{i}(p)_{C}$, these logics seem to be related. Even so, not all the connectives (especially the negation) received a satisfactory interpretation in this way. In addition, we might worry how such instrumentalist pictures of truth may get in the way of a more realist perspective on quantum theory.

\subsection{A Topos as an Intuitionistic Universe of Sets}

A topos is a mathematical universe of discourse. But the axiom of choice and the law of excluded middle may very well lead to contradictions in such a universe. This gives possibilities which are not allowed in the topos $\mathbf{Set}$. As an example there is a topos ($Sh(\mathbf{T})$, where $\mathbf{T}$ is a site of topological spaces with the open cover topology) such that all functions $\underline{\mathbb{R}}_{d}\to\underline{\mathbb{R}}_{d}$ are continuous~\cite[Section VI.9]{mm}. 

But are any of these new possibilities relevant to physics? And how impractical is it to lose the axiom of choice? In analysis, a branch of mathematics used extensively throughout physics, this axiom plays an important role. Consider the following quotation, taken from the well-known text~\cite{ped} on functional analysis.

\begin{quote}
This means that our acceptance of the axiom of choice determines what sort of mathematics we want to create, and it may in the end affect our mathematical description of physical realities. The same is true (albeit on a smaller scale) with the parallel axiom in euclidean geometry. But as the advocates of the axiom of choice, among them Hilbert and von Neumann, point out, several key results in modern mathematical analysis (e.g. the Tychonoff theorem, the Hahn-Banach theorem, the Krein-Milman theorem and Gelfand theory) depend crucially on the axiom of choice.
\end{quote}

Even so, we have used a topos valid version of Gelfand duality theory for the covariant model. This version of Gelfand theory~\cite{banmul3} does not rely on the axiom of choice. This works because locales take the place of topological spaces. In the same way, the other theorems mentioned by Pedersen have a localic counterpart which do not require the axiom of choice. At the risk of presenting things somewhat simpler than they really are, we might say that it is the emphasis on points in the notion of a topological space, that makes us need an axiom which generates enough points. With less emphasis on points, the axiom of choice becomes less powerful, possibly even superfluous.

Removing emphasis on points is relevant when we consider Isham's motivation for using topos theory in physics. In particular, consider the potential problem in quantum gravity of using real numbers as values for physical quantities, and, associated to that, the use of smooth manifolds for space-time. In~\cite{ish} we read:

\begin{quote}
Indeed, it is not hard to convince oneself that, from a physical perspective, the important ingredient in a space-time model is not the `points' in that space, but rather the `regions' in which physical entities can reside. In the context of a topological space, such regions are best modelled by open sets: the closed sets may be too `thin' to contain a physical entity, and the only physically meaningful closed sets are those with a non-empty interior. These reflections lead naturally to the subject of `pointless topology' and the theory of locales-a natural step along the road to topos theory.
\end{quote}

The arguments thus far only claimed that dropping the axiom of choice might not be as big a problem as one would expect at first. Thus far no arguments have shown that there is an actual advantage in dropping the axiom. In all honesty, the author does not know of any physically motivated arguments.

The extent to which the axiom of choice holds (in the internal language) depends on the topos~\cite{fs}. In the topos $\mathbf{Set}$ we can assume the full axiom of choice. For a category of presheaves $[\mathcal{C}^{op},\mathbf{Set}]$ (or copresheaves) the weaker version called the axiom of dependent choice can be assumed (by assuming the full axiom of choice for $\mathbf{Set}$), but for many categories $\mathcal{C}$, assuming the full axiom of choice leads to contradictions. If the topos is of the form $Sh(X)$, with $X$ a topological space, even the axiom of dependent choice may lead to contradictions. 

Assuming that the axiom of choice holds internally for a topos has large consequences for the topos. In particular, the topos is boolean, which means that the internal Heyting algebra $\Omega$ is an internal boolean algebra~\cite{bell2}. The contravariant model was founded in considerations of coarse-graining. In particular, the idea that truth values should correspond to down-closed subsets of $\mathcal{C}$ is a key ingredient. This leads to a topos which is not boolean, and therefore does not satisfy the (internal) axiom of choice. For the covariant model, using up-closed subsets of $\mathcal{C}$ as truth values does not seem a motivational point, but rather a consequence of the choice of topos. In~\cite{spit}, Spitters considers using the dense topology on the copresheaf topos $[\mathcal{C},\mathbf{Set}]$. The resulting topos of sheaves is a boolean topos, which satisfies the axiom of choice. Although it would be interesting to investigate this topos, we postpone further discussion to another paper.\\

One of the more striking possibilities granted by the absence of the law of excluded middle is synthetic differential geometry. Certain topoi~\cite{more} can act as models for differential geometry, using (nilpotent) infinitesimals. In the presence of the law of excluded middle all such infinitesimals would be forced to be equal to zero. As argued in detail in~\cite{bell}, (following Lawvere) these infinitesimals could allow us to deal with the continuum in a mathematical more natural way. At this point we might frown and say: thinking of quantum gravity we would like to get rid of the continuum rather than giving it a face-lift! Maybe so, but the problem of continuous vs. discrete in quantum gravity is deep and subtle. Having mathematical universes that capture the subtleties of the continuum may assist in understanding this problem better. Still, the current topos models $[\mathcal{C}^{op},\mathbf{Set}]$ and $[\mathcal{C},\mathbf{Set}]$ are not models for synthetic differential geometry. Speculating for a moment, it may be interesting to see if the categories used for studying locally covariant quantum field theories~\cite{brfrve}, can be extended to models of synthetic differential geometry, in such a way that the quantum field theories can be studied internally. \\

In the foundations of mathematics, dropping the law of excluded middle and axiom of choice is often motivated by a view on mathematical concepts as products of the mathematicians mind in favour of a view where mathematical concepts exist independent of us in some (platonic) realm. These very same mathematical concepts are used to represent ideas from physics. Of course, we should not confuse these mathematical representations with the physical ideas themselves, but, on the other hand, it can sometimes be hard to distinguish where the physics stops and the math starts (and vice versa). Modern physics is concerned with what we can say about nature, rather than what nature is (completely independent of us), a stance usually associated with Bohr. Under this premise, is it not more natural to represent the concepts of physics using abstractions which are thought of to originate from our mind rather than a platonic realm? But even if we concede this point, the relation between mathematics and physics is a complex one, therefore it remains to be seen whether or not a constructivist attitude with regards to mathematics has any physical significance, or if it is just a nuisance for the physicist whom is versed only in classical mathematics.
 
\section{Including $\ast$-homomorphisms} \label{sec: homs}

In this section we shift our attention from kinematics to dynamics in the topos models. In particular, we study how the spectral presheaf $\uS_{A}$ and the spectral locale $\uS_{\uA}$ transform under the action of a $\ast$-automorphism $h:A\to A$. Almost all of the material in this section relies on the ideas in~\cite{nuiten} or coincides with constructions from~\cite{doering, doering2}. From a distance the ideas in these references may appear different, but they are in fact closely related. In~\cite{nuiten} the emphasis is on the covariant model, and constructions such as daseinisation of self-adjoint operators are not considered. In~\cite{doering, doering2} the emphasis is on the contravariant model, and the internal perspective of the topos is not considered. Below we treat these ideas on dynamics for both topos models, with emphasis on internal reasoning.

\subsection{Covariant Model}

\subsubsection{C*-algebras}

As remarked in Section~\ref{sec: back}, the covariant model is typically applied to unital C*-algebras instead of the smaller class of von Neumann algebras. For the moment we will use all unital C*-algebras. In Subsection~\ref{subsec: dasdyn}, when daseinisation enters into the discussion, we will again restrict attention to von Neumann algebras. 

In the covariant approach, given a unital C*-algebra $A$, we assign to it a pair $([\mathcal{C}_{A},\mathbf{Set}],\underline{A})$, consisting of a topos and a unital commutative C*-algebra internal to this topos. In this section we look at the way $\ast$-homomorphisms $f:A\to B$ induce morphisms on the associated pairs $([\mathcal{C}_{A},\mathbf{Set}],\underline{A})$, $([\mathcal{C}_{B},\mathbf{Set}],\underline{B})$. We start by recalling two categories, introduced in earlier literature on topos approaches to quantum theory, and which will help in answering this question. We will subsequently show how these two categories are related. The first of the two is the category $\mathbf{cCTopos}_{N}$, introduced by Nuiten~\cite[Definition 4]{nuiten}.

\begin{dork}
The category $\mathbf{cCTopos}_{N}$ consists of the following
\begin{itemize}
\item \textit{Objects} are pairs $(\mathcal{E},\underline{A})$, where $\mathcal{E}$ is a topos and $\underline{A}$ is a unital commutative C*-algebra internal to the topos $\mathcal{E}$.
\item An \textit{arrow} $(G,\underline{g}):(\mathcal{E},\underline{A})\to(\mathcal{F},\underline{B})$, is given by a geometric morphism $G:\mathcal{E}\to\mathcal{F}$ and a $\ast$-homomorphism $\underline{g}:G^{\ast}\underline{B}\to\underline{A}$ in $\mathcal{F}$.
\item Composition of arrows is defined by $(G,\underline{g})\circ(F,\underline{f})=(G\circ F,\underline{f}\circ F^{\ast}\underline{g})$.
\end{itemize}
\end{dork}

For an arbitrary geometric morphism $G$, the object $G^{\ast}\underline{B}$ need not be a C*-algebra in $\mathcal{F}$. It is, at the very least, a semi-normed commutative $\ast$-algebra. The notion of a $\ast$-homomorphism, in the sense of a $\ast$-preserving homomorphism of algebras, still makes sense when the domain is $G^{\ast}\underline{B}$. If the geometric morphism comes from an $\ast$-homomorphism, as discussed below, then $G^{\ast}\underline{B}$ will always be an internal C*-algebra. Otherwise, we can take its Cauchy completion and turn it into an internal C*-algebra.

The second category of interest was introduced by Andreas D\"oring in~\cite{doering}:

\begin{dork}
Let $\mathcal{C}$ be any small category. The category $\mathbf{Copresh}(\mathcal{C})$ is defined by
\begin{itemize}
\item \textit{Objects} are functors $Q:\mathcal{J}\to\mathcal{C}$, where $J$ is any small category.
\item An \textit{arrow} $(f,\eta):Q_{1}\to Q_{2}$, where $Q_{i}:\mathcal{J}_{i}\to\mathcal{C}$, is given by a functor $f:\mathcal{J}_{1}\to\mathcal{J}_{2}$, and a natural transformation $\eta: Q_{1}\to f^{\ast}Q_{2}$. Here $f^{\ast}$ denotes the inverse image functor of the essential geometric morphism associated to $f$.
\end{itemize}
\end{dork}

The motivating example is when $\mathcal{C}$ is equal to $\mathbf{cuC^{\ast}}$, the category of unital commutative C*-algebras and unit preserving $\ast$-homomorphisms. Using presheaf semantics, one can prove that functors $A:\mathcal{J}\to\mathbf{ucC^{\ast}}$ correspond exactly to the unital commutative C*-algebra internal to the topos $[\mathcal{J},\mathbf{Set}]$. We can think of the objects of $\mathbf{Copresh}(\mathbf{ucC^{\ast}})$ as pairs $(\mathcal{E},\underline{A})$, where $\mathcal{E}$ is a topos (and in particular a functor category), and $\underline{A}$ is a unital commutative C*-algebra in $\mathcal{E}$.

For any pair of small categories $\mathcal{J}_{1}$, $\mathcal{J}_{2}$, a functor $f:\mathcal{J}_{1}\to\mathcal{J}_{2}$ defines an essential geometric morphism $F:[\mathcal{J}_{1},\mathbf{Set}]\to[\mathcal{J}_{2},\mathbf{Set}]$, where essential means that the inverse image functor $F^{\ast}$ (which is the left adjoint in the adjunction defining $F$) also has a left adjoint $F_{!}$. See, for example, \cite[(VII.2 Theorem 2)]{mm}. The following lemma gives a converse of this statement.

\begin{lem}{(\cite{jh1} Lemma 4.1.5)}
Let $\mathcal{J}_{1}$ and $\mathcal{J}_{2}$ be two small categories such that $\mathcal{J}_{2}$ is Cauchy-complete (i.e., all idempotent morphisms split). Then every essential geometric morphism $[\mathcal{J}_{1},\mathbf{Set}]\to[\mathcal{J}_{2},\mathbf{Set}]$ is induced by a functor $\mathcal{J}_{1}\to\mathcal{J}_{2}$ as above.
\end{lem}

If the base category $\mathcal{D}$ is a poset, then the only idempotent arrows are the identity morphisms. The base categories for the quantum topoi are therefore trivially Cauchy-complete. On the level of contexts, the order-preserving maps $\phi:\mathcal{C}_{A}\to\mathcal{C}_{B}$, correspond to geometric morphisms between the corresponding topoi, where the left-adjoint $\phi^{\ast}$ itself has a left adjoint $\phi_{!}$.

An arrow in $\mathbf{Copresh}(\mathbf{ucC^{\ast}})$ can thus be seen as a pair $(F,\underline{f}):(\mathcal{E},\underline{A})\to(\mathcal{F},\underline{B})$, where the topoi $\mathcal{E}$, $\mathcal{F}$ are functor categories, $F:\mathcal{E}\to\mathcal{F}$ is an essential geometric morphism, and $\underline{f}:\underline{A}\to F^{\ast}\underline{B}$ is a natural transformation. 

We replace the category $\mathbf{Copresh}(\mathbf{ucC^{\ast}})$ by the a related category.

\begin{dork}
The category $\mathbf{cCTopos}_{D}$ is given by:
\begin{itemize}
\item Objects are pairs $(\mathcal{E},\underline{A})$, where $\mathcal{E}$ is a topos and $\underline{A}$ is a unital commutative C*-algebra in $\mathcal{E}$.
\item Arrows $(F,\underline{f}):(\mathcal{E},\underline{A})\to(\mathcal{F},\underline{B})$ are given by a geometric morphism $F:\mathcal{E}\to\mathcal{F}$, and a $\ast$-homomorphism $\underline{f}:\underline{A}\to F^{\ast}\underline{B}$ in $\mathcal{E}$.
\item Composition of arrows is defined by $(G,\underline{g})\circ(F,\underline{f})=(G\circ F, F^{\ast}\underline{g}\circ\underline{f})$.
\end{itemize}
\end{dork}

\begin{rem}
The category $\mathbf{Copresh}(\mathcal{C})$ was introduced in~\cite{doering} in connection with another category $\mathbf{Presh}(\mathcal{D})$, to which it is dual equivalent. Here $\mathcal{D}$ is a category which is dual equivalent to $\mathcal{C}$ by assumption. When we replace $\mathbf{Copresh}(\mathbf{ucC^{\ast}})$ by $\textbf{cCTopos}_{D}$ this duality is lost. In the next section, where we look at the contravariant version of the topos approach, a category closely connected to $\mathbf{Presh}(\mathcal{D})$ is considered.
\end{rem}

Let $f:A\to B$ be a unit-preserving $\ast$-homomorphism (in $\mathbf{Set}$). Then $f$ induces an arrow
\begin{equation}
(F,\underline{f}):([\mathcal{C}_{A},\mathbf{Set}],\underline{A})\to([\mathcal{C}_{B},\mathbf{Set}],\underline{B})
\end{equation}
in $\mathbf{cCTopos}_{D}$. To see this, observe that $f$ induces an order-preserving map
\begin{equation}
\hat{f}:\mathcal{C}_{A}\to\mathcal{C}_{B},\ \ \hat{f}(C)=f[C],
\end{equation}
which in turn induces a geometric morphism $F:[\mathcal{C}_{A},\mathbf{Set}]\to[\mathcal{C}_{B},\mathbf{Set}]$. The inverse image functor acting on $\underline{B}$ is given by
\begin{equation}
F^{\ast}\underline{B}:\mathcal{C}_{A}\to\mathbf{Set},\ F^{\ast}\underline{B}(C)=\underline{B}\circ\hat{f}(C)=f[C].
\end{equation}
The internal $\ast$-homomorphism induced by $f$ is now simply given by
\begin{equation}
\underline{f}:\underline{A}\to F^{\ast}\underline{B},\ \ \ \underline{f}_{C}:C\to f[C],\ \ \ \underline{f}_{C}=f|_{C}.
\end{equation}

\begin{dork}
A unital $\ast$-homomorphism $f:A\to B$ is said to \textbf{reflect commutativity} if
\begin{equation*}
\forall a_{1},a_{2}\in A,\ \ [f(a_{1}),f(a_{2})]=0\ \Rightarrow\ [a_{1},a_{2}]=0.
\end{equation*}
\end{dork}

Note that if $f$ is injective, then $f$ reflects commutativity. A unital $\ast$-homomorphism $f:A\to B$ that reflects commutativity defines an arrow
\begin{equation}
(G,\underline{g}):([\mathcal{C}_{B},\mathbf{Set}],\underline{B})\to([\mathcal{C}_{A},\mathbf{Set}],\underline{A})
\end{equation}
in $\mathbf{cCTopos}_{N}$. As $f$ reflects commutativity we can define the order preserving map
\begin{equation}
\hat{g}:\mathcal{C}_{B}\to\mathcal{C}_{A},\ \ \hat{g}(D)=f^{-1}(D).
\end{equation}
As before, this induces an essential geometric morphism $G:[\mathcal{C}_{B},\mathbf{Set}]\to[\mathcal{C}_{A},\mathbf{Set}]$. The associated $\ast$-morphism is given by
\begin{equation}
\underline{g}:G^{\ast}\underline{A}\to\underline{B},\ \ \ \underline{g}_{D}: f^{-1}(D)\to D,\ \ \ \underline{g}_{D}=f|_{f^{-1}(D)}.
\end{equation}

Note that $\hat{g}$ is a right adjoint to $\hat{f}$. As a consequence, the geometric morphisms $F^{\ast}\dashv F_{\ast}$ and $G^{\ast}\dashv G_{\ast}$ are closely related. More precisely, $G_{\ast}=F^{\ast}$. As inverse image functors preserve colimits, it is clear that in this setting $\ast$-homomorphisms $G^{\ast}\underline{A}\to\underline{B}$ are equivalent to $\ast$-homomorphisms $\underline{A}\to F^{\ast}\underline{B}$.

\begin{rem}
In~\cite{nuiten} the $\ast$-homomorphisms under consideration are all inclusions, which obviously reflect commutativity. In general it is an open question which $\ast$-homomorphisms reflect commutativity, and which do not.
\end{rem}

We end with a small summary of the material in this section.

\begin{poe}
A unital $\ast$-homomorphism $f:A\to B$ induces an arrow
\begin{equation}
(F,\underline{f}):([\mathcal{C}_{A},\mathbf{Set}],\underline{A})\to([\mathcal{C}_{B},\mathbf{Set}],\underline{B})
\end{equation}
in $\mathbf{cCTopos}_{D}$ such that the internal $\ast$-homomorphism $\underline{f}$ is given by $\underline{f}_{C}=f|_{C}$. If $f$ reflects commutativity, then it also induces an arrow
\begin{equation}
(G,\underline{g}):([\mathcal{C}_{B},\mathbf{Set}],\underline{B})\to([\mathcal{C}_{A},\mathbf{Set}],\underline{A})
\end{equation}
in $\mathbf{cCTopos}_{N}$ such that the internal $\ast$-homomorphism $\underline{g}$ is $\underline{g}_{D}=f|_{f^{-1}(D)}$.
\end{poe}

\subsubsection{Locales}

In this subsection we describe the internal $\ast$-homomorphisms of the previous subsection at the level of the Gelfand spectra. We use the following observations. As noted before, given a locale $X$, in $\mathbf{Set}$, the categories  $\mathbf{Loc}_{Sh(X)}$ and $\mathbf{Loc}/X$ are equivalent~\cite[C1.6]{jh1}. In addition, a map of locales $f:X\to Y$ induces an adjunction
\begin{equation} \label{equ: localeadjunction}
(F_{\ast}\dashv F^{\sharp})\ \ \ F_{\ast}: \mathbf{Loc}_{Sh(X)}\to\mathbf{Loc}_{Sh(Y)}: F^{\sharp}.
\end{equation}
There is a good reason for writing the left adjoint as $F_{\ast}$. The continuous map $f$ defines a geometric morphism $F:Sh(X)\to Sh(Y)$. Unlike the inverse image functor $F^{\ast}$, the direct image functor $F_{\ast}$ preserves frames and morphisms of frames. In fact, this observation is crucial for the equivalence of the categories  $\mathbf{Loc}_{Sh(X)}$ and $\mathbf{Loc}/X$. The left adjoint $F_{\ast}$ of (\ref{equ: localeadjunction}) is the restriction of the direct image functor $F_{\ast}$ to frames and frame homomorphisms. The right adjoint $F^{\sharp}$, which should not be confused with the inverse image functor $F^{\ast}$, is most easily described under the identification $\mathbf{Loc}_{Sh(X)}\cong\mathbf{Loc}/X$. As a functor $\mathbf{Loc}/Y\to\mathbf{Loc}/X$ it maps a bundle $Z\to Y$, to the pullback of this bundle along the map $f:X\to Y$.

In \cite{nuiten}, in addition to $\mathbf{cCTopos}_{N}$, another, related, category was introduced.

\begin{dork}
The category $\mathbf{spTopos}_{N}$ of spaced topoi is given by
\begin{itemize}
\item Objects are pairs $(\mathcal{E},\underline{L})$, where $\mathcal{E}$ is a topos and $\underline{L}$ is a locale in $\mathcal{E}$.
\item An arrow $(G,\underline{s}):(\mathcal{E},\underline{L})\to(\mathcal{F},\underline{M})$ is given by a geometric morphism $G:\mathcal{E}\to\mathcal{F}$ and a locale map $\underline{s}:G_{\ast}\underline{L}\to\underline{M}$ in $\mathcal{F}$.
\item Composition of arrows is defined as $(G,\underline{t})\circ(F,\underline{s})=(G\circ F,\underline{t}\circ G_{\ast}\underline{s})$.
\end{itemize}
\end{dork}

A unital C*-algebra $A$ defines a spaced topos $([\mathcal{C}_{A},\mathbf{Set}],\uS_{\underline{A}})$, where $\uS_{\underline{A}}$ denotes the internal Gelfand spectrum of $\underline{A}$. A unital $\ast$-homomorphism $f:A\to B$ that reflects commutativity defines an arrow in $\mathbf{spTopos}_{N}$ as follows. We know that $f$ induces a $\ast$-homomorphism $\underline{g}:G^{\ast}\underline{A}\to\underline{B}$. By Gelfand duality this defines a locale map on the spectra
\begin{equation*}
\uS(\underline{g}):\uS_{\underline{B}}\to\uS_{G^{\ast}\underline{A}}.
\end{equation*}

Recall from Proposition~\ref{prop: extspecco} that the spectrum $\underline{\Sigma}_{\underline{B}}$ can be described externally as the bundle of topological spaces
\begin{equation}
\pi_{B}:\Sigma^{\uparrow}_{\underline{B}}\to\mathcal{C}^{\uparrow}_{B},\ \ (D,\lambda)\mapsto D,
\end{equation}
Analogously, the spectrum $\uS_{G^{\ast}\underline{A}}$ can be represented by the bundle
\begin{equation}
\Sigma_{G^{\ast}\underline{A}}^{\uparrow}\to\mathcal{C}^{\uparrow}_{B},\ \ (D,\lambda)\mapsto D,
\end{equation}
where, as sets, $\Sigma_{G^{\ast}\underline{A}}^{\uparrow}$ is equal to $\coprod_{D\in\mathcal{C}_{B}}\Sigma_{f^{-1}(D)}$, and $U$ is open in  $\Sigma_{G^{\ast}\underline{A}}^{\uparrow}$ iff the following two conditions hold:
\begin{enumerate}
\item If $D\in\mathcal{C}_{B}$, then $U_{D}:=U\cap\Sigma_{f^{-1}(D)}$ is open in $\Sigma_{f^{-1}(D)}$;
\item If $D_{1}\subseteq D_{2}$, then $\rho_{f^{-1}(D_{2})f^{-1}(D_{1})}^{-1}(U_{D_{1}})\subseteq U_{D_{2}}$, where $\rho_{f^{-1}(D_{2})f^{-1}(D_{1})}:\Sigma_{f^{-1}(D_{2})}\to\Sigma_{f^{-1}(D_{1})}$ is the 
restriction map.
\end{enumerate}

A straightforward calculation (or \cite[Lemma 3.4]{nuiten}) reveals that this bundle is simply the bundle $\pi_{A}:\Sigma_{\underline{A}}^{\uparrow}\to\mathcal{C}^{\uparrow}_{A}$, pulled back along the order-preserving function $\hat{g}:\mathcal{C}^{\uparrow}_{B}\to\mathcal{C}^{\uparrow}_{A}$, $D\mapsto f^{-1}(D)$, seen as an Alexandroff-continuous map. Externally, the map $\uS(\underline{g})$, is given by
\[ \xymatrix{ \Sigma^{\uparrow}_{\underline{B}} \ar[rr]^{\Sigma(\underline{g})} \ar[dr]_{\pi_{B}} & & \hat{g}^{\ast}\Sigma^{\uparrow}_{\underline{A}} \ar[dl]^{\hat{g}^{\ast}\pi_{A}} \\
& \mathcal{C}^{\uparrow}_{B} &} \]
\begin{equation} \label{equ: nuitv}
\Sigma(\underline{g}):\Sigma^{\uparrow}_{\underline{B}}\to\hat{g}^{\ast}\Sigma^{\uparrow}_{\underline{A}},\ \ (D,\lambda)\mapsto (D,\lambda\circ f|_{f^{-1}(D)}).
\end{equation}

Note that internally this is a locale map $\uS_{\underline{B}}\to G^{\sharp}\uS_{\underline{A}}$ in $[\mathcal{C}_{B},\mathbf{Set}]$. This is, in turn, equivalent to a locale map $G_{\ast}\uS_{\underline{B}}\to\uS_{\underline{A}}$ in $[\mathcal{C}_{A},\mathbf{Set}]$. In this way, the $\ast$-homomorphism $f: A\to B$ defines a morphism
\begin{equation*}
([\mathcal{C}_{B},\mathbf{Set}],\uS_{\underline{B}})\to([\mathcal{C}_{A},\mathbf{Set}],\uS_{\underline{A}})
\end{equation*}
in the category $\mathbf{spTopos}_{N}$.

Next, drop the assumption that $f:A\to B$ reflects commutativity. From the previous subsection we know that $f$ defines a $\ast$-homomorphism $\underline{f}:\underline{A}\to F^{\ast}\underline{B}$ in $[\mathcal{C}_{A},\mathbf{Set}]$. As before, by Gelfand duality this defines a continuous map of locales on the spectra
\begin{equation}
\uS(\underline{f}):\underline{\Sigma}_{F^{\ast}\underline{B}}\to\underline{\Sigma}_{\underline{A}}.
\end{equation}

The spectrum $\uS_{F^{\ast}\underline{B}}$ can be represented by the bundle
\begin{equation} \label{equ: doebundle}
\Sigma_{F^{\ast}\underline{B}}^{\uparrow}\to\mathcal{C}^{\uparrow}_{A},\ \ (C,\lambda)\mapsto C,
\end{equation}
where, as sets, $\Sigma_{F^{\ast}\underline{B}}^{\uparrow}$ is equal to $\coprod_{C\in\mathcal{C}_{A}}\Sigma_{f[C]}$, and $U$ is open in  $\Sigma_{F^{\ast}\underline{B}}^{\uparrow}$ iff the following two conditions hold:
\begin{enumerate}
\item If $C\in\mathcal{C}_{A}$, then $U_{C}:=U\cap\Sigma_{f[C]}$ is open in $\Sigma_{f[C]}$;
\item If $C_{1}\subseteq C_{2}$, then $\rho_{f[C_{2}]f[C_{1}]}^{-1}(U_{C_{1}})\subseteq U_{C_{2}}$, where $\rho_{f[C_{2}]f[C_{1}]}:\Sigma_{f[C_{2}]}\to\Sigma_{f[C_{1}]}$ is the restriction map.
\end{enumerate}

This bundle can be identified as $\pi_{B}:\Sigma^{\uparrow}_{\underline{B}}\to\mathcal{C}^{\uparrow}_{B}$, pulled back along $\hat{f}:\mathcal{C}_{A}^{\uparrow}\to\mathcal{C}^{\uparrow}_{B}$, $C\mapsto f[C]$. Externally, the locale map $\uS(\underline{f}):\uS_{F^{\ast}\underline{B}}\to\uS_{\underline{A}}$ is given by the continuous function
\begin{equation} \label{equ: d-morph bundle}
\Sigma(\underline{f}):\hat{f}^{\ast}\Sigma^{\uparrow}_{\underline{B}}\to\Sigma^{\uparrow}_{\underline{A}},\ \ (C,\lambda)\mapsto (C,\lambda\circ f|_{C}),
\end{equation}
over $\mathcal{C}^{\uparrow}_{A}$. Note that in (\ref{equ: d-morph bundle}) $\lambda\in\Sigma_{f[C]}$. Internally we obtain a locale map $F^{\sharp}\uS_{\underline{B}}\to\uS_{\underline{A}}$.

For the remainder of this subsection, assume once again that $f$ reflects commutativity. How is the locale map $\Sigma(\underline{f}): F^{\sharp}\uS_{\underline{B}}\to\uS_{\underline{A}}$, obtained from the $\ast$-homomorphism $\underline{f}:\underline{A}\to F^{\ast}\underline{B}$ related to the locale map\footnote{Using $\Sigma(\underline{g})$ for this map is a slight abuse of notation, as this name was used earlier to denote the corresponding locale map $\uS_{\underline{B}}\to G^{\sharp}\uS_{\underline{A}}$.} $\Sigma(\underline{g}): G_{\ast}\uS_{\underline{B}}\to\uS_{\underline{A}}$ obtained from the $\ast$-homomorphism $\underline{g}: G^{\ast}\underline{A}\to\underline{B}$? We know that $G_{\ast}= F^{\ast}$, but this does not imply that on the level of locales $G_{\ast}$ and $F^{\sharp}$ are the same. In fact, $G_{\ast}\uS_{\underline{B}}$ and $F^{\sharp}\uS_{\underline{B}}$ are slightly different locales.

The locale $\uS_{\underline{B}}$ corresponds to a frame object $\mathcal{O}\uS_{\underline{B}}$ in the topos $[\mathcal{C}_{B},\mathbf{Set}]$. For $C\in\mathcal{C}_{A}$,
\begin{equation} \label{berek}
G_{\ast}(\mathcal{O}\uS_{\underline{B}})(C)=F^{\ast}(\mathcal{O}\uS_{\underline{B}})(C)=\mathcal{O}\uS_{\underline{B}}(f[C]).
\end{equation}
Using the external description $\Sigma^{\uparrow}_{\underline{B}}$, the right-hand side of (\ref{berek}) is given by the subspace topology of $\Sigma^{\uparrow}_{\underline{B}}$ on the subset $\coprod_{D\in\mathcal{C}_{B}\cap(\uparrow f[C])}\Sigma_{D}$.
\begin{equation} \label{berek2}
G_{\ast}(\mathcal{O}\uS_{\underline{B}})(C)=\mathcal{O}\left(\coprod_{D\in\mathcal{C}_{B}\cap(\uparrow f[C])}\Sigma_{D}\right).
\end{equation}
On the other hand
\begin{equation} \label{berek3}
\mathcal{O}(F^{\sharp}\uS_{\underline{B}})(C)=\mathcal{O}\left(\coprod_{C'\in\mathcal{C}_{A}\cap(\uparrow C)}\Sigma_{f[C]}\right)
\end{equation}
where, on the right-hand side we take the subspace topology from $\hat{f}^{\ast}\Sigma_{\underline{B}}^{\uparrow}$. We can now see that the sets (\ref{berek2}) and (\ref{berek3}) are different. The only difference is that $G_{\ast}(\mathcal{O}\uS_{\underline{B}})(C)$ considers all contexts $D\in\mathcal{C}_{B}$ which are above $f[C]$, whereas $\mathcal{O}(F^{\sharp}\uS_{\underline{B}})(C)$ only considers those contexts which come from an $C'\in\mathcal{C}_{A}$. As the locale map $\Sigma(\underline{g})$ comes from a $\ast$-homomorphism in $Sh(\mathcal{C}^{\uparrow}_{\underline{B}})$, and the locale map $\Sigma(\underline{f})$ comes from a $\ast$-homomorphism in $Sh(\mathcal{C}^{\uparrow}_{\underline{A}})$ this difference was to be expected.

\begin{poe} \label{prop: flowco}
A unital $\ast$-homomorphism $f:A\to B$ induces a continuous map of locales $\uS(\underline{f}):F^{\sharp}\uS_{\underline{B}}\to\uS_{\underline{A}}$ in $[\mathcal{C}_{A},\mathbf{Set}]$. The external description of this map is given by the continuous function 
\begin{equation*}
\Sigma(\underline{f}):\hat{f}^{\ast}\Sigma^{\uparrow}_{\underline{B}}\to\Sigma^{\uparrow}_{\underline{A}},\ \ (C,\lambda)\mapsto(C,\lambda\circ f|_{C}). 
\end{equation*}
If $f$ reflects commutativity, then there is also a locale map $\uS(\underline{g}):\uS_{\underline{B}}\to G^{\sharp}\uS_{\underline{A}}$ in $[\mathcal{C}_{B},\mathbf{Set}]$ externally given by (\ref{equ: nuitv}).
\end{poe}

If we think of $\uS_{\uA}$ as an internal state space, then ideally a $\ast$-automorphism $h:A\to A$ induces an isomorphism of locales $\uS_{\uA}\to\uS_{\uA}$ internal to the topos. However, the automorphism $h$ induces a map $\hat{h}:\mathcal{C}_{A}\to\mathcal{C}_{A}$, and we need to take into account how $h$ shuffles the contexts around. Instead of an isomorphism $\uS_{\uA}\to\uS_{\uA}$ we arrived at an isomorphism of locales of the form $H^{\sharp}\uS_{\underline{A}}\to\uS_{\underline{A}}$.

\subsection{Contravariant Version}

In the contravariant model, we associate a pair $([\mathcal{C}_{A}^{op},\mathbf{Set}],\uS_{A})$ to a von Neumann algebra $A$, consisting of a topos and a topological space in this topos. Here $\uS_{A}$ is the spectral presheaf, equipped with the internal topology generated by the closed open subobjects. Motivated by the locale maps of the previous subsection, we show that a unital $\ast$-homomorphism $f:A\to B$ induces a pair
\begin{equation*}
(F,\uS(\underline{f})):([\mathcal{C}^{op}_{A},\mathbf{Set}],\uS_{A})\to([\mathcal{C}^{op}_{B},\mathbf{Set}],\uS_{B})
\end{equation*}
consisting of a geometric morphism $F:[\mathcal{C}^{op}_{A},\mathbf{Set}]\to[\mathcal{C}^{op}_{B},\mathbf{Set}]$ and a continuous map $\uS(\underline{f}):F^{\ast}\uS_{B}\to\uS_{A}$ in the topos $[\mathcal{C}^{op}_{A},\mathbf{Set}]$. The first question which we need to address is how the object $F^{\ast}\uS_{B}$ is an internal topological space.

As an object $F^{\ast}\uS_{B}$, is constructed as follows. The functor $\uS_{B}$ can be described as an \'etale bundle $\pi_{B}:\SdB\to\mathcal{C}^{\downarrow}_{B}$. As a set, $\SdB$ is equal to $\coprod_{D\in\mathcal{C}_{B}}\Sigma_{D}$, and $U\subseteq\SdB$ is open iff
\begin{equation} \label{equ: etalecontra}
\text{If}\ \ (D,\lambda)\in U\ \ \text{and}\ \ D'\subseteq D,\ \ \text{then}\ \ (D',\lambda|_{D'})\in U.
\end{equation}
As an \'etale bundle $F^{\ast}\uS_{B}$ is the pullback of the \'etale bundle $\pi_{B}$ along $\hat{f}:\mathcal{C}^{\downarrow}_{A}\to\mathcal{C}^{\downarrow}_{B}$, $\hat{f}(C)=f[C]$. The bundle $\hat{f}^{\ast}\pi_{B}:\hat{f}^{\ast}\SdB\to\mathcal{C}^{\downarrow}_{A}$ obtained in this way can be described as follows. As a set, the total space $\hat{f}^{\ast}\SdB$ is equal to $\coprod_{C\in\mathcal{C}_{A}}\Sigma_{f[C]}$. A subset $U\subseteq\hat{f}^{\ast}\SdB$ is open iff it satisfies the following condition: if, for $C\in\mathcal{C}_{A}$, $\lambda\in\Sigma_{f[C]}$, $(C,\lambda)\in U$, and $C'\subseteq C$ in $\mathcal{C}_{A}$, then $(C',\lambda|_{C'})\in U$. The map $\hat{f}^{\ast}\pi_{B}$ is simply $(C,\lambda)\mapsto C$.

The internal topology on $\uS_{B}$ corresponds to a topology on $\SdB$, which is coarser than the \'etale topology, but with respect to which $\pi_{B}$ is still continuous. With respect to this topology $U\in\mathcal{O}\SdB$ iff it is \'etale open in $\SdB$, and, in addition, for each $D\in\mathcal{C}_{B}$, the set $U_{D}:=U\cap\Sigma_{D}$ is open in $\Sigma_{D}$. We can take the pullback of $\pi_{B}$ along $\hat{f}$, with this new topology on $\SdB$, and obtain a coarser topology on $\hat{f}^{\ast}\SdB$ than the \'etale topology. In fact $U\in\mathcal{O}\SuB$ iff it is \'etale open and, for each $C\in\mathcal{C}_{A}$, $U_{C}=U\cap\Sigma_{f[C]}$ is open in $\Sigma_{f[C]}$. The bundle $\hat{f}^{\ast}\pi_{B}:\hat{f}^{\ast}\SdB\to\mathcal{C}^{\downarrow}_{A}$ is continuous with respect to this new topology. We have thus defined an internal topology on $F^{\ast}\uS_{B}$. It is the topology generated by the objects $F^{\ast}\underline{U}$, where $\underline{U}$ is a closed open subobject of $\uS_{B}$. Whenever we consider $F^{\ast}\uS_{B}$ as a topological space, it is with respect to this topology.

Now that we have identified $F^{\ast}\uS_{B}$ as an internal topological space, we can define the function $\uS(\underline{f})$ and check whether it is continuous.

\begin{poe} \label{prop: flowcontra}
The natural transformation $\uS(\underline{f}):F^{\ast}\uS_{B}\to\uS_{A}$, given by
\begin{equation}
\uS(\underline{f})_{C}: \Sigma_{f[C]}\to\Sigma_{C},\ \ \lambda\mapsto\lambda\circ f|_{C}
\end{equation}
is a continuous map of topological spaces in $[\mathcal{C}^{op}_{A},\mathbf{Set}]$.
\end{poe}

\begin{proof}
We leave the verification that $\uS(\underline{f})$ is indeed a natural transformation to the reader. At the level of \'etale bundles, $\uS(\underline{f})$ corresponds to the commuting triangle
\[ \xymatrix{ \hat{f}^{\ast}\Sigma^{\downarrow}_{B} \ar[rr]^{\Sigma(\underline{f})} \ar[dr]_{\hat{f}^{\ast}\pi_{B}} & & \Sigma^{\downarrow}_{B} \ar[dl]^{\pi_{A}} \\
& \mathcal{C}^{\downarrow}_{A} &} \]
of continuous maps, where the total spaces of the bundles are equipped with the \'etale topologies. Note that naturality of $\uS(\underline{f})$ amounts to continuity of $\Sigma(\underline{f})$ with respect to the \'etale topologies. Also note that, as functions, the function $\Sigma(\underline{f})$ is the same function as (\ref{equ: d-morph bundle}) from the covariant version. The only difference between the approaches is in the topologies.

The function $\uS(\underline{f})$ is internally continuous iff $\Sigma(\underline{f})$ is also continuous with respect to the coarser topologies on the total spaces (corresponding to the internal topologies). Let $\Sigma(f|_{C}):\Sigma_{f[C]}\to\Sigma_{C}$ be the Gelfand dual of the $\ast$-homomorphism $f|_{C}:C\to f[C]$. A straightforward check reveals that for any $U\in\mathcal{O}\SdA$, and $C\in\mathcal{C}_{A}$
\begin{equation} \label{etalecont}
\Sigma(\underline{f})^{-1}(U)_{C}=\Sigma(f|_{C})^{-1}(U_{C})\in\mathcal{O}\Sigma_{f[C]}.
\end{equation}
This observation combined with \'etale continuity proves that $\Sigma(\underline{f})$ is continuous with respect to the desired topologies. Note that \'etale continuity can be deduced from (\ref{etalecont}), as for $\lambda\in\Sigma_{f[C]}$, and $C'\subseteq C$,  clearly $(\lambda\circ f|_{C})|_{C'}=\lambda|_{C'}\circ f|_{C'}$.
\end{proof}

Proposition~\ref{prop: flowcontra} is the contravariant counterpart to Proposition~\ref{prop: flowco}. A $\ast$-automorphism $h:A\to A$ induces a homeomorphism $\uS(\underline{h}):H^{\ast}\uS_{A}\to\uS_{A}$. In the following section we consider how elementary propositions $\underline{[a\in\Delta]}$ transform under the frame isomorphism $\uS(\underline{h})^{-1}$.

\begin{rem}
If we ignore the \'etale topology of $\uS_{B}$ and consider it to be a locale rather than an internal space, then the bundle map $\Sigma(\underline{f})$, from the previous proof, can be seen as an internal locale map $\uS(\underline{f}):F^{\sharp}\uS_{B}\to\uS_{A}$, as in the covariant case.
\end{rem}

\begin{rem}
As in the previous subsection, if $f:A\to B$ reflects commutativity, we can define a continuous map of spaces $\uS(\underline{g}):\uS_{B}\to G^{\ast}\uS_{A}$ in the topos $[\mathcal{C}_{B}^{op},\mathbf{Set}]$, or see it as a locale map $\uS(\underline{g}):\uS_{B}\to G^{\sharp}\uS_{A}$ in the same topos.
\end{rem}

\subsection{Automorphisms and Daseinisation} \label{subsec: dasdyn}

Let $A$ be a von Neumann algebra and $h:A\to A$ a $\ast$-automorphism. In this section we investigate how daseinised self-adjoint operators transform under $h$. We will be working with the contravariant version. However, if we switch from internal spaces to locales, switch inner and outer daseinisation, replace $\uS_{A}$ by $\uS_{\uA}$, switch order-reversing and order-preserving, and replace $\downarrow$ by $\uparrow$ whenever it occurs as a superscript, then this section is about the covariant version instead.

For $a\in A_{sa}$, outer daseinisation defines a continuous map $\underline{\delta^{o}(a)}:\uS_{A}\to\underline{\mathbb{R}}_{l}$, and inner daseinisation  a continuous map $\underline{\delta^{i}(a)}:\uS_{A}\to\underline{\mathbb{R}}_{u}$, where $\underline{\mathbb{R}}_{l}$ and $\underline{\mathbb{R}}_{u}$ are the spaces of lower and upper reals respectively. From the previous section we know that $h$ induces a continuous map $\uS(\underline{h}):H^{\ast}\uS_{A}\to\uS_{A}$. We can compose these maps to obtain continuous maps
\begin{equation*}
\underline{\delta^{o}(a)}_{h}:H^{\ast}\uS_{A}\to\underline{\mathbb{R}}_{l},\ \ \ \underline{\delta^{i}(a)}_{h}:H^{\ast}\uS_{A}\to\underline{\mathbb{R}}_{u}.
\end{equation*}

If we look at~\cite[Section 10]{di}, we may suspect that there is a relation between $\underline{\delta^{o}(h(a))}$ and $\underline{\delta^{o}(a)}_{h}$ and also betwenn their inner counterparts. This is indeed the case, and we will proceed to describe this.

\begin{lem} \label{lem: specproj}
Let $a\in A_{sa}$, and $\Delta$ a Borel subset of $\sigma(a)$, the spectrum of $a$. Then 
\begin{equation*}
h(\chi_{\Delta}(a))=\chi_{\Delta}(h(a)).
\end{equation*}
\end{lem}

\begin{proof}
First note that $\sigma(h(a))=\sigma(a)$. If $p(a)$ denotes a polynomial in $a$ with complex coefficients, then $h(p(a))=p(h(a))$. By norm-continuity of $h$ and the Stone-Weierstrass Theorem, $h$ restricts to an isomorphism of unital C*-algebras
\begin{equation*}
\tilde{h}:C^{\ast}(a,1)\to C^{\ast}(h(a),1).
\end{equation*}
Let $W^{\ast}(a)=C^{\ast}(a,1)''$ denote the weak as well as the $\sigma$-weak closure of $C^{\ast}(a)$. Any $\ast$-automorphism is $\sigma$-weakly continuous, implying that $\tilde{h}$ extends to an isomorphism of abelian von Neumann algebras
\begin{equation*}
\tilde{h}:W^{\ast}(a)\to W^{\ast}(h(a))
\end{equation*}
As shown in~\cite{dixm} there exists an isomorphism of von Neumann algebras $i:L^{\infty}(\sigma(a),\mu)\to W^{\ast}(a)$, where $\mu$ is any scalar-valued spectral measure on $\sigma(a)$. For $W^{\ast}(h(a))$ we can construct an isomorphism $j: L^{\infty}(\sigma(a),\mu)\to W^{\ast}(h(a))$, using the same $\mu$, since the spectra of $a$ and $h(a)$ coincide. With these identifications we obtain an automorphism $\hat{h}=j^{-1}\circ\tilde{h}\circ i$ of the abelian von Neumann algebra $L^{\infty}(\sigma(a),\mu)$. By construction, for each polynomial expression $p(x):\sigma(a)\to\mathbb{C}$, we deduce $\hat{h}([p(x)])=[p(x)]$. By $\sigma$-weak continuity, $\hat{h}$ is the identity map. The desired claim follows from
\begin{equation*}
h(\chi_{\Delta}(a))=\tilde{h}(\chi_{\Delta}(a))=\tilde{h}(i([\chi_{\Delta}]))=j([\chi_{\Delta}])=\chi_{\Delta}(h(a)).
\end{equation*}
\end{proof}

\begin{lem}
If $a\leq_{s} b$ in $A_{sa}$, then $h(a)\leq_{s} h(b)$ in $A_{sa}$.
\end{lem}

\begin{proof}
Let $a\leq_{s} b$ in $A_{sa}$, let $(e^{a}_{x})_{x\in\mathbb{R}}$ be the spectral resolution of $a$, and $(e^{b}_{x})_{x\in\mathbb{R}}$ the spectral resolution of $b$. By assumption, for each $x\in\mathbb{R}$, $e^{b}_{x}\leq e^{a}_{x}$. The family of projections $e^{h(a)}_{x}:=h(e^{a}_{x})$ defines a spectral resolution for $h(a)$. This can be verified as $h$, restricted to the projections of $A$, yields an isomorphism of complete lattices. Alternatively, it follows straight from the previous lemma. Likewise, $e^{h(b)}_{x}:=h(e^{b}_{x})$ is a spectral resolution for $h(b)$. Any $\ast$-homomorphism $h$ is a positive map, so, from the assumption, we deduce that for each $x\in\mathbb{R}$, $e^{h(b)}_{x}\leq e^{h(a)}_{x}$. We conclude that $h(a)\leq_{s} h(b)$.
\end{proof}

\begin{cor} \label{cor}
If $a\in A_{sa}$, and $C\in\mathcal{C}_{A}$, then
\begin{equation}
h(\delta^{o}(a)_{C})=\delta^{o}(h(a))_{h[C]},\ \ \ h(\delta^{i}(a)_{C})=\delta^{i}(h(a))_{h[C]}.
\end{equation}
\end{cor}

\begin{proof}
By the previous lemma the bijection $h|_{A_{sa}}:A_{sa}\to A_{sa}$ is monotone with respect to the spectral order. It has a order-preserving inverse, making it an order-isomorphism. As a consequence $h|_{A_{sa}}$ is an isomorphism of boundedly complete lattices. 
\begin{align*}
h(\delta^{o}(a)_{C}) &= h\left(\bigwedge\{b\in C_{sa}\mid b\geq_{s} a\}\right) \\
&= \bigwedge\{h(b)\in h[C_{sa}]\mid b\geq_{s} a\} \\
&= \bigwedge\{ h(b)\in h[C]_{sa}\mid h(b)\geq_{s} h(a)\} \\
&= \bigwedge\{ c\in h[C]_{sa}\mid c\geq_{s} h(a)\} \\
&= \delta^{o}(h(a))_{h[C]} 
\end{align*}
Inner daseinisation can be treated in the same way.
\end{proof}

\noindent The continuous map $\underline{\delta^{o}(a)}_{h}:H^{\ast}\uS_{A}\to\underline{\mathbb{R}}_{l}$ is externally described by the triangle  of continuous maps
\[ \xymatrix{ \hat{h}^{\ast}\SdA \ar[dr]_{\hat{h}^{\ast}\pi_{A}} \ar[rr]^{\delta^{o}(a)_{h}} & &\mathcal{R}_{l} \ar[dl]^{\pi_{l}} \\
& \mathcal{C}^{\downarrow}_{A} &} \]
where the elements of $\mathcal{R}_{l}$ are pairs $(C,s)$, with $C\in\mathcal{C}_{A}$ and $s:(\downarrow C)\to\mathbb{R}$ is an order-reversing function. For $\lambda\in\Sigma_{h[C]}$
\begin{equation*}
\delta^{o}(a)_{h}(C,\lambda):\downarrow C\to\mathbb{R},\ \ D\mapsto\langle\delta^{o}(a)_{D},\lambda\circ h|_{C}\rangle.
\end{equation*}
Note that
\begin{align*}
\langle\delta^{o}(a)_{D},\lambda\circ h|_{C}\rangle &= \langle\delta^{o}(a)_{D},\Sigma(h|_{C})(\lambda)\rangle\\
&= \langle(h|_{C})(\delta^{o}(a)_{D}),\lambda\rangle\\
&= \langle\delta^{o}(h(a))_{h[D]},\lambda\rangle, 
\end{align*}
where, in the last step, we used Corollary~\ref{cor}.

We need one more definition before we can state the relations we are looking for. Define the continuous map of spaces
\begin{equation*}
\underline{\mathbb{R}}_{l}(\underline{h}): H^{\ast}\underline{\mathbb{R}}_{l}\to\underline{\mathbb{R}}_{l},
\end{equation*}
\begin{equation*}
\underline{\mathbb{R}}_{l}(\underline{h})_{C}: \text{OR}(\downarrow h[C],\mathbb{R})\to\text{OR}(\downarrow C,\mathbb{R}), \ \ s\mapsto s\circ\hat{h}|_{\downarrow C}.
\end{equation*}

\begin{poe}
Let $h:A\to A$ be a $\ast$-automorphism, and take $a\in A_{sa}$. Then the following square of continuous maps of spaces is commutative:
\[ \xymatrix{ H^{\ast}\uS_{A} \ar[rr]^{H^{\ast}(\underline{\delta^{o}(h(a))})} \ar[d]_{\uS(\underline{h})} & & H^{\ast}\underline{\mathbb{R}}_{l} \ar[d]^{\underline{\mathbb{R}}_{l}(\underline{h})}\\
\uS_{A} \ar[rr]_{\underline{\delta^{o}(a)}} & & \underline{\mathbb{R}}_{l} }\]
The same holds for inner daseinisation if we replace $\underline{\mathbb{R}}_{l}$ by $\underline{\mathbb{R}}_{u}$.
\end{poe} 

Next, we look at the action of the automorphism on the elementary propositions. The map $\uS(\underline{h}):H^{\ast}\uS\to\uS$ is continuous, providing us with an inverse image map $\uS(\underline{h})^{-1}:\mathcal{O}\uS\to\mathcal{O}H^{\ast}\uS$. Let $\underline{[a\in\Delta]}$ be the elementary proposition obtained by outer daseinisation of the spectral projection $\chi_{\Delta}(a)$. To describe the open $\uS(\underline{h})^{-1}(\underline{[a\in\Delta]})$ of $\mathcal{O}H^{\ast}\uS$ it is convenient to take the external descriptions. 

For a $\lambda\in\Sigma_{h[C]}$, by definition $(C,\lambda)\in\Sigma(\underline{h})^{-1}([a\in\Delta])$ iff $\lambda\circ h|_{C}\in [a\in\Delta]_{C}$. This happens iff
\begin{equation*}
1=\langle\delta^{o}(\chi_{\Delta}(a))_{C},\lambda\circ h|_{C}\rangle=\langle h(\delta^{o}(\chi_{\Delta}(a))_{C}),\lambda\rangle.
\end{equation*}
This can be simplified further using Lemma~\ref{lem: specproj}
\begin{equation*}
h(\delta^{o}(\chi_{\Delta}(a))_{C})=\delta^{o}(h(\chi_{\Delta}(a)))_{h[C]}=\delta^{o}(\chi_{\Delta}(h(a)))_{h[C]}.
\end{equation*}
We conclude
\begin{equation*}
\Sigma(\underline{h})^{-1}([a\in\Delta])=\{(C,\lambda)\in\hat{h}^{\ast}\Sigma_{\downarrow}\mid\langle\delta^{o}(\chi_{\Delta}(h(a)))_{h[C]},\lambda\rangle=1\}.
\end{equation*}

Note that if $\underline{[a\in\Delta]}=\underline{\delta(a)}^{-1}(\underline{\Delta})$, and 
\begin{equation*}
\underline{\delta(a)}_{h}=\langle\underline{\delta^{i}(a)}_{h},\underline{\delta^{o}(a)}_{h}\rangle:H^{\ast}\uS\to\underline{\mathbb{R}}_{u}\times\underline{\mathbb{R}}_{l},
\end{equation*}
then
\begin{equation*}
\underline{\Sigma}(\underline{h})^{-1}(\underline{[a\in\Delta]})=\underline{\delta(a)}_{h}^{-1}(\underline{\Delta}).
\end{equation*}

We would like to combine this open with a state, seen as a probability valuation, in order to obtain a truth value. There is one problem however, as the open lies in $\mathcal{O}\hat{h}^{\ast}\Sigma_{\downarrow}$ and not in $\mathcal{O}\Sigma_{\downarrow}$. Recall that a state $\psi$ defines a probability valuation by combining the probability valuations $\mu^{C}:\mathcal{O}\Sigma_{C}\to[0,1]$, corresponding to the local states $\psi|_{C}:C\to\mathbb{C}$. Externally viewed, the valuation $\underline{\mu}$ is given by
\begin{equation*}
\mu:\mathcal{O}\Sigma_{\downarrow}\to\text{OR}(\mathcal{C},[0,1])\ \ \ \mu(U)(C)=\mu^{C}(U_{C}).
\end{equation*}
In very much the same way, a state $\psi$ defines the function
\begin{equation*}
\mu_{h}:\mathcal{O}\hat{h}^{\ast}\Sigma_{\downarrow}\to\text{OR}(\mathcal{C},[0,1])\ \ \ \mu_{h}(U)(C)=\mu^{h[C]}(U_{C}).
\end{equation*}
Note that for $U\in\mathcal{O}\hat{h}^{\ast}\Sigma_{\downarrow}$, $U_{C}\in\mathcal{O}\Sigma_{h[C]}$. The reader is invited to check that $\mu_{h}$ satisfies all conditions required to turn the corresponding internal function $\underline{\mu}_{h} :\mathcal{O}H^{\ast}\uS\to\underline{[0,1]}_{l}$ into a probability valuation.

Using $\underline{\mu}_{h}$ we can once again obtain truth values. Since daseinisation and automorphisms interact well, we obtain the following result.

\begin{tut}
The following two forcing conditions are equivalent:
\begin{enumerate}
\item $C\Vdash\underline{\mu}_{h}(\underline{\Sigma}(\underline{h})^{-1}(\underline{[a\in\Delta]}))=\underline{1},$
\item $h[C]\Vdash\underline{\mu}(\underline{[h(a)\in\Delta]})=\underline{1}.$
\end{enumerate}
\end{tut}

\begin{proof}
Spelling out the first condition gives
\begin{equation*}
\mu^{h[C]}(\{\lambda\in\Sigma_{h[C]}\mid\langle\delta^{o}(\chi_{\Delta}(h(a)))_{h[C]},\lambda\rangle=1\})=1
\end{equation*}
or, equivalently
\begin{equation*}
\psi(\delta^{o}(\chi_{\Delta}(h(a)))_{h[C]})=1,
\end{equation*}
which is just the second forcing relation of the proposition.
\end{proof}

As daseinisation of self-adjoint operators commutes with $\ast$-automorphisms, the elementary propositions of both the covariant and contravariant approach transform in a simple way under the action of $\uS(\underline{h})^{-1}$. As a consequence, the theorem given above states that if $S$ is the cosieve or sieve of the proposition $\uS(\underline{h})^{-1}(\underline{[a\in\Delta]})$ relative to some state $\psi$, then $h[S]$ is the cosieve or sieve of $\underline{[h(a)\in\Delta]}$ relative to that same state $\psi$.

\section{Acknowledgements}

The author would like to thank Klaas Landsman, Filip Bar, Matthijs V\'ak\'ar, Jeremy Butterfield, Hans Halvorson and Ronnie Hermens for sharing their thoughts on some of the subjects discussed above.


\begin{thebibliography}{99}
\bibitem{ab} Samson Abramsky and Adam Brandenburger. The Sheaf-Theoretic Structure of Non-Locality and Contextuality. http://arxiv.org/abs/1102.0264v7. 2011.
\bibitem{banmul3} Bernhard Banaschewski and Christopher J. Mulvey. A Globalisation of the Gelfand Duality Theorem. Ann. Pure Appl. Logic. 137(1-3), 62-103. 2006.
\bibitem{bell2} John L. Bell. Toposes and Local Set Theories. An Introduction. Dover Publications Inc. 1988.  
\bibitem{bell} John L. Bell. A Primer in Infinitesimal Analysis, 2nd Edition. Cambridge University Press. 2008.
\bibitem{bor} Francis Borceux. Handbook of Categorical Algebra 3: Categories of Sheaves. Encyclopedia of Mathematics and its Applications 52. Cambridge University Press. 1994.
\bibitem{brfrve} Romeo Brunetti, Klaus Fredenhagen and Rainer Verch. The Generally Covariant Locality Principle - A New Paradigm for Local Quantum Field Theory. Commun. Math. Phys .237. p31-68. 2003.
\bibitem{buhais} Jeremy Butterfield, John Hamilton and Chris Isham. A Topos Perspective on the Kochen-Specker Theorem: III. von Neumann algebras as the Base Category. Int. J. Theor. Phys. 39, 1413-1436. 2000.
\bibitem{butish1} Jeremy Butterfield and Chris Isham. A Topos Perspective on the Kochen-Specker Theorem: I. Quantum States as Generalized Valuations. Int. J. Theor. Phys. 37(11),2669-2733. 1998.
\bibitem{butish2} Jeremy Butterfield and Chris Isham. A Topos Perspective on the Kochen-Specker Theorem: II. Conceptual Aspects and Classical Analogues. Int. J. Theor. Phys. 38(3), 827-859. 1999.
\bibitem{butish3} Jeremy Butterfield and Chris Isham. A Topos Perspective on the Kochen-Specker Theorem: IV. Interval Valuations. Int. J. Theor. Phys 41, 613-639. 2002.  
\bibitem{chls} Martijn Caspers, Chris Heunen, Nicolaas P. Landsman, and Bas Spitters. Intuitionistic Quantum Logic of an $n$-Level System. Foundations of Physics, 39(7):731-759. ArXiv: 0902.3201v2. 2009.
\bibitem{coq} Thierry Coquand. About Stone's Notion of Spectrum. Journal of Pure and Applied Algebra,197: 141-158. 2005.
\bibitem{dixm} Jacques Dixmier. Von Neumann Algebras. North-Holland Publishing Company. 1981.
\bibitem{doe3} Andreas D\"oring. Quantum States and Measures on the Spectral Presheaf.  Adv. Sci. Lett. on Quantum Gravity, Cosmology and Black Holes, ed. M. Bojowald. ArXiv:0809.4847v. 2008
\bibitem{doering} Andreas D\"oring. Generalized Gelfand Spectra of Nonabelian Unital C*-Algebras 1: Categorical Aspects, Automorphisms and Jordan Structure. ArXiv: 1212.2613. 2012.
\bibitem{doering2} Andreas D\"oring. Generalized Gelfand Spectra of Nonabelian Unital C*-Algebras 2: Flows and Time Evolution of Quantum Systems. arXiv:1212.4882. 2012.
\bibitem{di1} Andreas D\"oring and Chris Isham. A Topos Foundation for Theories of Physics: I. Formal Languages for Physics. J. Math. Phys. 49, Issue 5, 053515. arXiv:quant-ph/0703060. 2008
\bibitem{di2} Andreas D\"oring and Chris Isham. A Topos Foundation for Theories of Physics: II. Daseinisation and the Liberation of Quantum Theory. J. Math. Phys. 49, Issue 5, 053516. arXiv:quant-ph/0703062. 2008
\bibitem{di3} Andreas D\"oring and Chris Isham. A Topos Foundation for Theories of Physics: III. Quantum Theory and the Representation of Physical Quantities with Arrows $\delta(A)$ J. Math. Phys. 49, Issue 5, 053517. arXiv:quant-ph/0703064. 2008.
\bibitem{di4} Andreas D\"oring and Chris Isham. A Topos Foundation for Theories of Physics: IV. Categories of Systems. J. Math. Phys. 49, Issue 5, 053518. arXiv:quant-ph/0703066. 2008.
\bibitem{di} Andreas D\"oring and Chris Isham. What is a Thing? New Structures for Physics, Lecture Notes in Physics. Springer. 2009.
\bibitem{di5} Andreas D\"oring and Chris Isham. Classical and Quantum Probabilities as Truth Values. arXiv:1102.2213. 2011.
\bibitem{fs} Michael P. Fourman and Andre Scedrov. the World's Simplest Axiom of Choice Fails. Manuscripta Mathematica 38, 325-332. 1982.
\bibitem{gol} Robert Goldblatt. Topoi: the Categorical Analysis of Logic. Amsterdam: North-Holland. 1984.
\bibitem{hado} John Harding and Andreas D\"oring. Abelian Subalgebras and the Jordan Structure of a von Neumann Algebra. arXiv:1009.4945. 2010
\bibitem{hls} Chris Heunen, Nicolaas P. Landsman and Bas Spitters. A Topos for Algebraic Quantum Theory. Communications in Mathematical Physics 291. 63-110. 2009.
\bibitem{ish} Chris Isham. Topos Methods in the Foundations of Physics. Deep Beauty: Mathematical Innovation and Research for Underlying Intelligibility in the Quantum World. Cambridge University Press. 2011. 
\bibitem{jh1} Peter T. Johnstone. Sketches of an Elephant: A Topos Theory Compendium. Oxford University Press. 2002.
\bibitem{jh4} Peter T. Johnstone. Stone Spaces. Cambridge University Press. 1982.
\bibitem{kari} Richard V. Kadison and John R. Ringrose. Fundamentals of the Theory of Operator Algebras. Academic Press. 1983.
\bibitem{mm} Saunders Mac Lane and Ieke Moerdijk. Sheaves in Geometry and Logic; A First Introduction to Topos Theory. Springer. 1992.
\bibitem{moerdijk} Ieke Moerdijk. Spaced Spaces. Compositio Mathematica, 53(2). p171-209. 1984
\bibitem{more} Ieke Moerdijk and Gonzalo E. Reyes. Models for Smooth Infinitesimal Analysis. Springer-Verlag. 1991.
\bibitem{nuiten} Joost Nuiten. Bohrification of Local Nets of Observables. arXiv:1109.1397v1. 2011.
\bibitem{ols} Milton P. Olson. The Selfadjoint Operators of a von Neumann Algebra form a Conditionally Complete Lattice. Proc. of the AMS 28, 537-544. 1971.
\bibitem{ped} Gert K. Pedersen. Analysis Now. Graduate Texts in Mathematics 118. Springer-Verlag. 1995.
\bibitem{pipu} Jorge Picado and Ale\v s Pultr. Frames and Locales, Topology without Points.  2010.
\bibitem{spit} Bas Spitters. The Space of Measurement Outcomes as a Spectral Invariant for Non-Commutative Algebras. Foundations of Physics. 2011.
\bibitem{vic} Steve Vickers. Locales and Toposes as Spaces. Handbook of Spatial Logics. 429-496. Springer. 2007.
\bibitem{vic2} Steve Vickers. A Localic Theory of Lower and Upper Integrals. Mathematical Logic Quarterly. 54(1) p109-123. 2008.
\bibitem{wollie} Sander A.M. Wolters. A Comparison of Two Topos-Theoretic Approaches to Quantum Theory. Commun. Math. Phys 317(1) p3-53. 2013.
\end{thebibliography}
\end{document}